\documentclass[10pt,letterpaper,conference]{IEEEtran}

\usepackage[T1]{fontenc} % optional
\usepackage{times}
\usepackage{amsmath}
\usepackage{amssymb}
\usepackage{amsthm}
\usepackage{color}
\usepackage{algorithm}
\usepackage[noend]{algorithmic}
\usepackage{subfigure}
\usepackage{multirow}
\usepackage[bookmarks=false,colorlinks=false,pdfborder={0 0 0}]{hyperref}
\usepackage{cite}
\usepackage{bm}
\usepackage{arydshln}
\usepackage{mathtools}
\usepackage{microtype}
\usepackage{subfigure}
\usepackage{float}
\usepackage{makecell}
\usepackage{graphicx}
\usepackage{graphics}

\usepackage{multicol}   
\usepackage{threeparttable}  
\usepackage{rotating}
\usepackage{mathrsfs}
\usepackage{authblk}
\usepackage{colortbl}

\usepackage{makecell,multirow,diagbox}

\newtheorem{theorem}{Theorem}

\newtheorem{lemma}[theorem]{Lemma}

\long\def\symbolfootnote[#1]#2{\begingroup
	\def\thefootnote{\fnsymbol{footnote}}\footnote[#1]{#2}\endgroup}
\renewcommand{\paragraph}[1]{{\bf #1}}

\long\def\symbolfootnote[#1]#2{\begingroup
	\def\thefootnote{\fnsymbol{footnote}}\footnote[#1]{#2}\endgroup}

\ifodd1%revise of the text
\newcommand{\com}[1]{\textbf{\color{blue} (COMMENT: #1)}}%comment of the text
\else\newcommand{\com}[1]{}\fi

\begin{document}
	\title{Joint Design of Piggyback and Conjugate Transformation Functions for Repair Bandwidth Reduction in Piggybacking Codes}
	
% 	\author{Hao Shi$^\dagger$, Zhengyi Jiang$^\dagger$$^\ddagger$, Zhongyi Huang$^\dagger$,
% Bo Bai$^\ddagger$, Gong Zhang$^\ddagger$, and Hanxu Hou$^\ddagger$$^\star$\\
% $^\dagger$ Department of Mathematics Sciences, Tsinghua University, Beijing, China \\
% % $^\ddagger$ Theory Lab, Huawei Tech. Co. Ltd., Beijing, China
% $^\ddagger$ Theory Lab, Central Research Institute, 2012 Labs, Huawei Tech. Co. Ltd., Hong Kong SAR
% }
\author{Hao Shi, Zhengyi Jiang, Gefeng Deng, Zhongyi Huang, and Hanxu Hou}

	\maketitle
	
	\begin{abstract}
 \symbolfootnote[0]{
This paper was presented in part at the IEEE Information Theory Workshop (ITW), 2024 \cite{Shi2411:Conjugate}.
H. Shi, Z. Jiang, and Z. Huang are with 
the Department of Mathematics Sciences, Tsinghua University, Beijing, China~(E-mail: shih22@mails.tsinghua.edu.cn, jzy10492761962@163.com, zhongyih@tsinghua.edu.cn). G. Deng is with the
Shenzhen Institute for Advanced Study, University of Electronic Science and Technology of China, Chengdu, China~(E-mail: 202512281042@std.uestc.edu.cn). 
H. Hou is with the Shenzhen University of Advanced Technology~(E-mail: houhanxu@163.com). 

This research is supported in part by National Key Research and Development Program of China under Grant No. 2025YFA1017200.  {\em (Corresponding author: Hanxu Hou.)}}
Efficient node repair is a central requirement in distributed storage systems, particularly in high-rate erasure-coded deployments where repair traffic directly affects network overhead and recovery cost. Piggybacking codes reduce the repair bandwidth of MDS array codes while keeping the sub-packetization level small. However, existing piggybacking constructions often rely on restrictive piggyback-function designs to preserve the MDS property over small fields, which limits their repair-bandwidth reduction.
We propose {\em conjugate-piggybacking} codes, a new class of MDS array codes that jointly design piggyback functions and conjugate transformations under small sub-packetization. The proposed construction improves repair efficiency while preserving the MDS property over moderate field sizes. In particular, it enables some parity nodes to achieve optimal repair bandwidth and reduces the overall repair bandwidth compared with existing piggybacking-based designs.
We analyze the MDS property and repair bandwidth of the proposed codes and evaluate them against existing piggybacking codes under high-code-rate settings over $\mathbb{F}_{2^8}$. We further conduct a repair-traffic simulation under uniform single-node failures to quantify the expected traffic reduction in storage-oriented settings. The results show that our construction consistently achieves lower repair bandwidth than related piggybacking codes and reduces expected repair traffic compared with conventional RS repair. These gains are obtained at the cost of a slightly larger field size, revealing a practical trade-off between repair efficiency and field-size overhead for high-rate distributed storage.	\end{abstract}
	
	
	\IEEEpeerreviewmaketitle
\section{Introduction}
\label{sec:1}
Erasure coding is widely used in modern distributed storage systems to provide fault tolerance with substantially lower storage overhead than replication. Among erasure-coding schemes, maximum distance separable (MDS) codes are especially attractive because they provide the strongest reliability guarantee for a given redundancy budget. In practice, however, the efficiency of an erasure-coded storage system is determined not only by storage overhead but also by the cost of repairing failed nodes. This issue is particularly important in high-rate deployments, where repair traffic can significantly affect network overhead, recovery latency, and overall system performance.

Consider a storage system with $n$ storage nodes that encodes $\mathcal{M}=k\ell$ source symbols into $n\ell$ coded symbols, where each node stores $\ell$ symbols and $r=n-k$ denotes the number of parity nodes. If the original data can be reconstructed from any $k$ of the $n$ nodes, then the code has the $(n,k)$ MDS property, and we refer to it as an $(n,k,\ell)$ MDS array code. The parameter $\ell$ is called the sub-packetization level. In this paper, we focus on systematic MDS array codes with $k$ data nodes and $r$ parity nodes. Reed--Solomon (RS) codes \cite{reed1960} are the classical example, corresponding to the case $\ell=1$.

A major limitation of traditional MDS codes in storage applications is their high repair cost. To repair one failed node, a conventional MDS code typically downloads $\ell$ symbols from each of the $k$ surviving nodes, resulting in a repair bandwidth of $k\ell$ symbols. Regenerating codes, introduced by Dimakis {\em et al}. \cite{dimakis2010}, showed that the minimum repair bandwidth for an $(n,k,\ell)$ MDS array code is $\frac{(n-1)\ell}{n-k}$, and codes achieving this bound are known as minimum storage regenerating (MSR) codes. Although many MSR constructions have been developed \cite{rashmi2011,tamo2013,hou2016,2017Explicit,li2018,2018A,hou2019a,hou2019b}, their practical use in high-rate storage systems is often constrained by the rapid growth of sub-packetization. In particular, \cite{2018A} showed that, under high-code-rate parameters (i.e., $\frac{k}{n}>\frac{1}{2}$), the required sub-packetization level can grow exponentially with respect to $n$ and $k$.

Piggybacking codes provide an appealing alternative for storage systems that require lower repair traffic without the large sub-packetization overhead of MSR codes. Introduced by Rashmi {\em et al}. \cite{2017Piggybacking}, the piggybacking framework constructs MDS array codes with relatively small sub-packetization and modest field size while reducing the repair bandwidth for single-node failures compared with traditional MDS codes. Because of this favorable trade-off, piggybacking-based designs have received considerable attention in storage-oriented coding research \cite{2018Repair,2021piggybacking,2021piggyback,2019AnEfficient,2022Piggyback,2023gcpig,jiang2024toward}. However, under small-field and high-rate conditions, existing piggybacking codes typically require restrictive piggyback-function designs to preserve the MDS property, which limits the achievable reduction in repair bandwidth.

More broadly, several MDS array code constructions have been proposed to reduce repair bandwidth while keeping sub-packetization small \cite{2018hetc, 2023wk1,2023wk2,2023tran,2023gcplus}. This design space involves a fundamental trade-off among repair efficiency, field size, and sub-packetization. For example, HashTag Erasure Codes (HTEC) \cite{2018hetc}, Elastic Transformed-RS (ET-RS) codes \cite{2023tran}, and Bidirectional Piggybacking Design (BPD) \cite{2023wk1,2023wk2} can reduce repair bandwidth, but they require relatively large field sizes to maintain the MDS property. In contrast, piggybacking+ codes \cite{2023gcplus} operate over small fields of order $O(n)$, but they require the sub-packetization level to be a multiple of $r$ and strictly larger than $r$. Existing schemes therefore do not fully meet the need for a storage-efficient design that simultaneously supports low repair bandwidth, small sub-packetization, and moderate field size in high-rate regimes.

To address this problem, we propose {\em conjugate-piggybacking} codes, a new family of MDS array codes with sub-packetization $\ell=n-k$. The main idea is to jointly design piggyback functions and a conjugate transformation to improve repair efficiency under small-sub-packetization constraints. In the proposed construction, piggyback functions reduce the repair bandwidth of data nodes, while the conjugate transformation enables some parity nodes to achieve optimal repair bandwidth without increasing the sub-packetization level to a multiple of $r$. Compared with piggybacking+ codes \cite{2023gcplus}, our construction preserves the transformation benefit while keeping the sub-packetization fixed at $\ell=n-k$. Meanwhile, the MDS property can be maintained over a field of size $O(kr^2)$.

We analyze the MDS property and repair bandwidth of the proposed codes and show that they achieve lower repair bandwidth than existing related piggybacking codes. Under the evaluated high-code-rate setting with $r=4$, $4\leq k\leq 52$, $\ell=4$, and field $\mathbb{F}_{2^8}$, the proposed construction yields lower repair bandwidth than the existing related piggybacking schemes. We also conduct a repair-traffic simulation under representative high-rate parameters, showing that the proposed design reduces the expected single-node repair traffic by $37.0\%$--$54.8\%$ relative to conventional RS repair. These results indicate that conjugate-piggybacking codes provide a favorable trade-off among repair efficiency, sub-packetization, and field-size overhead for high-rate distributed storage systems.

% In Section \ref{sec:2}, we present the construction of our conjugate-piggybacking codes. In Section \ref{sec:3}, we present the repair method for our conjugate-piggybacking codes. In Section \ref{sec:4}, we analyze the repair bandwidth of conjugate-piggybacking codes. In Section \ref{sec:5}, we evaluate our conjugate-piggybacking codes and the existing piggybacking codes. We conclude the paper in Section \ref{sec:6}.
The remainder of this paper is organized as follows. 
Section~\ref{sec:1.1} presents an illustrative example of the proposed conjugate-piggybacking codes. Section~\ref{sec:2} gives the general construction. Section~\ref{sec:3} describes the repair procedures for single-node failures. Section~\ref{sec:4} analyzes the repair bandwidth of the proposed codes. Section~\ref{sec:5} evaluates the proposed construction and compares it with existing piggybacking-based schemes. Section~\ref{sec:6} concludes the paper.

\section{An Illustrative Example}
\label{sec:1.1}

We illustrate the main idea of the proposed conjugate-piggybacking design using an example with parameters $(n,k,\ell)=(14,10,4)$, as shown in Fig.~\ref{fig:0.1}. The example provides intuition for the construction and shows how the proposed design jointly improves the repair efficiency of data nodes and parity nodes under small sub-packetization.

\begin{figure*}[htpb]
	\centering
	\includegraphics[width=0.92\linewidth]{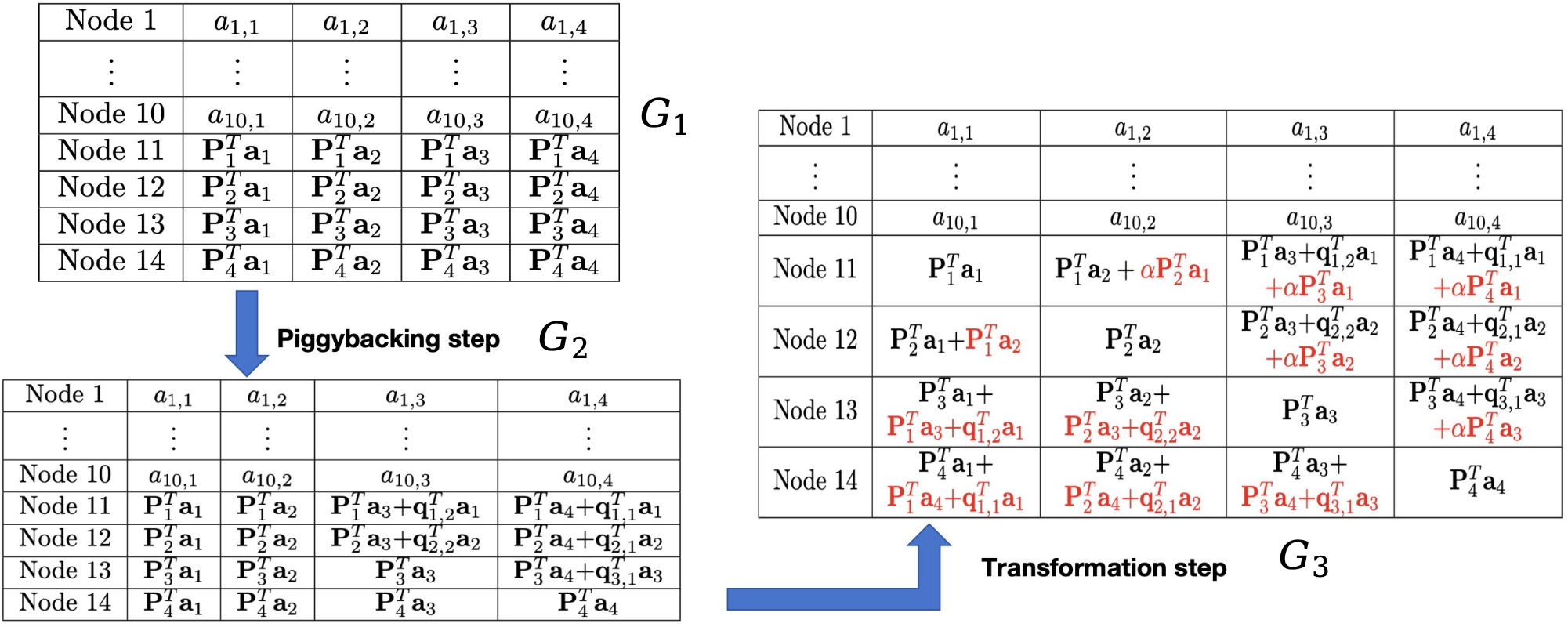}
	\caption{An example of $(n,k,\ell)=(14,10,4)$ conjugate-piggybacking codes. The array $\mathbf{G}_1$ consists of four instances of a $(14,10,1)$ systematic MDS code. The array $\mathbf{G}_2$ is obtained after the piggybacking step, and the array $\mathbf{G}_3$ is obtained after the transformation step. The red symbols in $\mathbf{G}_3$ are introduced by the transformation step.}
	\label{fig:0.1}
\end{figure*}

The construction proceeds in three steps.

First, we generate the $14\times 4$ array $\mathbf{G}_1$ from four $(14,10,1)$ systematic MDS code instances. For $i,j\in\{1,2,3,4\}$, let $\{a_{v,i}\}_{v=1}^{10}$ denote the $10$ data symbols in the $i$-th instance, let $\mathbf{a}_i=(a_{1,i},\ldots,a_{10,i})^T$, and let
\[
\mathbf{P}_j=(\alpha^j,\alpha^{2j},\ldots,\alpha^{10j})^T,
\]
where $\alpha$ is a primitive element of $\mathbb{F}_{2^8}$. Then
\[
(a_{1,i},\ldots,a_{10,i},\mathbf{P}_1^T\mathbf{a}_i,\ldots,\mathbf{P}_4^T\mathbf{a}_i)
\]
forms the $i$-th systematic MDS code instance.

Second, we add five piggyback functions to selected parity symbols in $\mathbf{G}_1$ to obtain $\mathbf{G}_2$. In this example, the piggyback functions are
\begin{eqnarray}
    &&\mathbf{q}_{1,1}=(\alpha,\alpha^2,\alpha^3,\alpha^4,0,0,0,0,0,0)^T,\nonumber\\
    &&\mathbf{q}_{1,2}=(0,0,0,0,\alpha^5,\alpha^6,\alpha^7,0,0,0)^T,\nonumber\\
    &&\mathbf{q}_{2,1}=(\alpha^2,\alpha^4,\alpha^6,\alpha^8,0,0,0,0,0,0)^T,\nonumber\\
    &&\mathbf{q}_{2,2}=(0,0,0,0,\alpha^{10},\alpha^{12},\alpha^{14},0,0,0)^T,\nonumber\\
    &&\mathbf{q}_{3,1}=(\alpha^3,\alpha^6,\alpha^9,\alpha^{12},0,0,0,0,0,0)^T.\nonumber
\end{eqnarray}

Third, we apply the conjugate transformation to all $16$ parity symbols of $\mathbf{G}_2$ to obtain the final array $\mathbf{G}_3$. By Theorem~\ref{th:1} in Section~\ref{sec:2.2}, the resulting code is MDS.

The array $\mathbf{G}_2$ remains a piggybacking code, whereas $\mathbf{G}_3$ does not. This is because the transformation from $\mathbf{G}_2$ to $\mathbf{G}_3$ couples symbols across all four parity nodes, while the conventional piggybacking framework only allows invertible transformations within an individual node \cite{2017Piggybacking}. This cross-parity coupling is the key mechanism that reduces the repair bandwidth of parity nodes in the proposed construction.

We next outline the repair process. Suppose that node $1$ in $\mathbf{G}_3$ fails. The repair first downloads $10$ symbols, namely $\{a_{i,4}\}_{i=2}^{10}$ and $\mathbf{P}_4^T\mathbf{a}_4$, to recover the erased symbol $a_{1,4}$ by the MDS property, and then obtains $\{\mathbf{P}_j^T\mathbf{a}_4\}_{j=1}^{3}$. For each $v\in\{1,2,3\}$, the repair further downloads two transformed parity symbols to solve for $\mathbf{P}_4^T\mathbf{a}_v$ and $\mathbf{P}_v^T\mathbf{a}_4+\mathbf{q}_{v,1}^T\mathbf{a}_v$, and finally downloads three systematic symbols $a_{2,v},a_{3,v},a_{4,v}$ to reconstruct the erased symbol $a_{1,v}$. Therefore, the total repair bandwidth of node $1$ is
\[
10+3\times(2+3)=25
\]
symbols. The same procedure gives a repair bandwidth of $25$ symbols for each of nodes $2$--$4$, $28$ symbols for each of nodes $5$--$7$, and $34$ symbols for each of nodes $8$--$10$.

Now suppose that parity node $11$ fails. The repair first downloads the $10$ symbols $\{a_{i,1}\}_{i=1}^{10}$ to recover $\mathbf{P}_1^T\mathbf{a}_1$ by the MDS property, and then obtains $\{\mathbf{P}_j^T\mathbf{a}_1\}_{j=2}^{4}$. It then downloads three additional parity expressions to reconstruct the remaining symbols in the failed parity node. The total repair bandwidth of node $11$ is therefore $13$ symbols. Similarly, the repair bandwidths of nodes $12$--$14$ are $13$, $19$, and $25$ symbols, respectively.

A key observation is that nodes $11$ and $12$ attain the MSR lower bound
\[
\frac{(n-1)\ell}{n-k}=\frac{13\times 4}{4}=13.
\]
This example illustrates the central design principle of the proposed code. The piggybacking step reduces the repair bandwidth of data nodes, while the conjugate transformation drives some parity nodes to optimal repair bandwidth. Their combination explains why conjugate-piggybacking codes can achieve lower repair bandwidth than existing related piggybacking constructions.

\section{Conjugate-piggybacking Codes}
\label{sec:2}
Throughout this paper, let $[u]:=\{1,2,\cdots,u\}$ for any positive integer $u$, and let $[u,v]:=\{u,u+1,\cdots,v\}$ for any positive integers $u<v$.
\subsection{Construction}
\label{sec:2.1}
The proposed conjugate-piggybacking codes are constructed in the following three steps. 

	\textbf{Step 1: Form array $\mathbf{G}_1$ containing $r$ basic systematic MDS code instances.} 
	
	 For positive integers $n,k$ with $n>k$, let $r=n-k$ and form an $n\times r$ array $\mathbf{G}_1$, where each row represents a storage node. Label the $n$ rows (or nodes) from 1 to $n$, and label the $r$ columns from 1 to $r$. Let $\{a_{i,j}\}_{i\in[k],j\in[r]}$ denote the $kr$ data symbols, where $a_{i,j}$ is placed in row $i$ and column $j$ of $\mathbf{G}_1$. Let $\mathbf{a}_i:=(a_{1,i},a_{2,i},\cdots,a_{k,i})^T$ for $i\in[r]$. 
  %we can see that $\mathbf{a}_i$ is the $k\times1$ data vector formed by all $k$ data symbols in the column $i$ of $\mathbf{G}_1$. 
	Let $\mathbf{P}_i=(\alpha^i,\alpha^{2i},\cdots,\alpha^{ki})^T$ for $i\in[r]$, where $\alpha$ is a primitive element of $\mathbb{F}_{2^q}$. 
 %Fill the rest of $\mathbf{G}_1$, specifically, for $i,j\in[r]$, place the symbol 
 The symbol $\mathbf{P}_i^T\mathbf{a}_j$ is placed in row $k+i$ and column $j$ of $\mathbf{G}_1$. For each $i\in[r]$, $(a_{1,i},a_{2,i},\cdots,a_{k,i},\mathbf{P}_1^T\mathbf{a}_i,\cdots,\mathbf{P}_r^T\mathbf{a}_i)$ is a codeword of an $(n,k,\ell=1)$ systematic MDS code. 
	
	\textbf{Step 2: Piggybacking step: form array $\mathbf{G}_2$ by designing piggybacking functions for data symbols.} 
	
	Let $L$ be a positive integer with $2\leq L\leq r$. Divide the $k$ data nodes into $L$ groups $\Phi_1,\Phi_2,\cdots,\Phi_L$ as evenly as possible, and let $n_i$ be the number of nodes in $\Phi_i$, for $i\in[L]$. Specifically, for $i\in[L]$, we take $$n_i=\begin{cases}
	\left \lceil \frac{k}{L} \right \rceil, & \text{ if } i\leq k-\left \lfloor \frac{k}{L} \right \rfloor L; \\ 
	\left \lfloor \frac{k}{L} \right \rfloor, & \text{ if } i>k-\left \lfloor \frac{k}{L} \right \rfloor L.
	\end{cases}$$

	For $i\in[r]$, 
% the $k\times 1$ vector $\mathbf{P}_i$ is also correspondingly divided into $L$ segments, i.e., we can write $\mathbf{P}_i$ as $\mathbf{P}_i=\sum_{t\in[L]}\mathbf{q}_{i,t}$, where the $k\times1$ vector 
let
 \begin{eqnarray}\nonumber
     &&\mathbf{q}_{i,t}=(0,0,\cdots,0,\alpha^{(\sum_{v=0}^{t-1}n_v+1)i},\alpha^{(\sum_{v=0}^{t-1}n_v+2)i},\\  \nonumber
     &&\cdots,\alpha^{(\sum_{v=0}^{t}n_v)i},0,0,\cdots,0)^T,\forall t\in[L], \nonumber
 \end{eqnarray}
	where $n_0:=0$.
 Then $\mathbf{P}_i=\sum_{t\in[L]}\mathbf{q}_{i,t}$.
	For $t\in[L]$, denote by $\mathbf{M}_t$ the rectangular region formed by the $n_{t}(r-t)$ data symbols in rows $(\sum_{v=0}^{t-1}n_v+1)$ to $(\sum_{v=0}^{t}n_v)$ and the first $r-t$ columns of $\mathbf{G}_1$.
	For $t\in[L-1]$, we design $r-t$ piggyback functions for the data symbols in $\mathbf{M}_t$ and add them to the $r-t$ parity symbols $\{\mathbf{P}_i^T\mathbf{a}_{r-t+1}\}_{i\in[r-t]}$ in column $r-t+1$ of $\mathbf{G}_1$. Specifically, for $j\in[r-t]$, the piggyback function $\mathbf{q}_{j,t}^T\mathbf{a}_{j}$ is added to the parity symbol $\mathbf{P}_{j}^T\mathbf{a}_{r-t+1}$. Let $\mathbf{G}_2$ denote the resulting array after the piggybacking step. In $\mathbf{G}_2$, denote by $\mathbf{R}$ the square region composed of all $r^2$ parity symbols stored in the $r$ parity nodes, and let $R(i,j)$ be the symbol in row $k+i$ and column $j$ of $\mathbf{R}$, for $i,j\in[r]$.
    The structure of $\mathbf{G}_2$ is shown in Fig.~\ref{fig:1}.% to clearly show the design idea of piggyback functions.
	
	\textbf{Step 3: Transformation step: form array $\mathbf{G}_3$ by replacing each pair of diagonally symmetric symbols in $\mathbf{R}$ with their linear combinations.}
	
	%On the basis of $\mathbf{G}_2$, we transform the pairwise parity symbols in region $\mathbf{R}$ which are located off the diagonal and symmetrical with the diagonal. 
 Specifically, for $i,j\in[r]$, define $P(i,j)$ as
	\begin{eqnarray}
		&&P(i,j)=\begin{cases}
			R(i,j)+\alpha\cdot R(j,i), & \text{ if } i<j; \\ 
			R(i,j), & \text{ if } i=j; \\ 
			R(i,j)+R(j,i), & \text{ if } i>j. 
		\end{cases}\label{eq:1}
	\end{eqnarray} 
 We replace each symbol $R(i,j)$ in $\mathbf{R}$ with $P(i,j)$ and denote the resulting array by $\mathbf{G}_3$. Note that for $i,j\in[r]$ with $i<j$, the pair $R(i,j),R(j,i)$ can be recovered from $P(i,j),P(j,i)$ as
		\begin{eqnarray}
		&&\begin{pmatrix}
			R(i,j)\\ 
			R(j,i)
		\end{pmatrix}=\begin{pmatrix}
			1 & \alpha\\ 
			1 & 1
		\end{pmatrix}^{-1}\begin{pmatrix}
			P(i,j)\\ 
			P(j,i)
		\end{pmatrix}.\label{eq:2}
	\end{eqnarray}

\begin{figure}[htpb]
	\centering
	\includegraphics[width=1.0\linewidth]{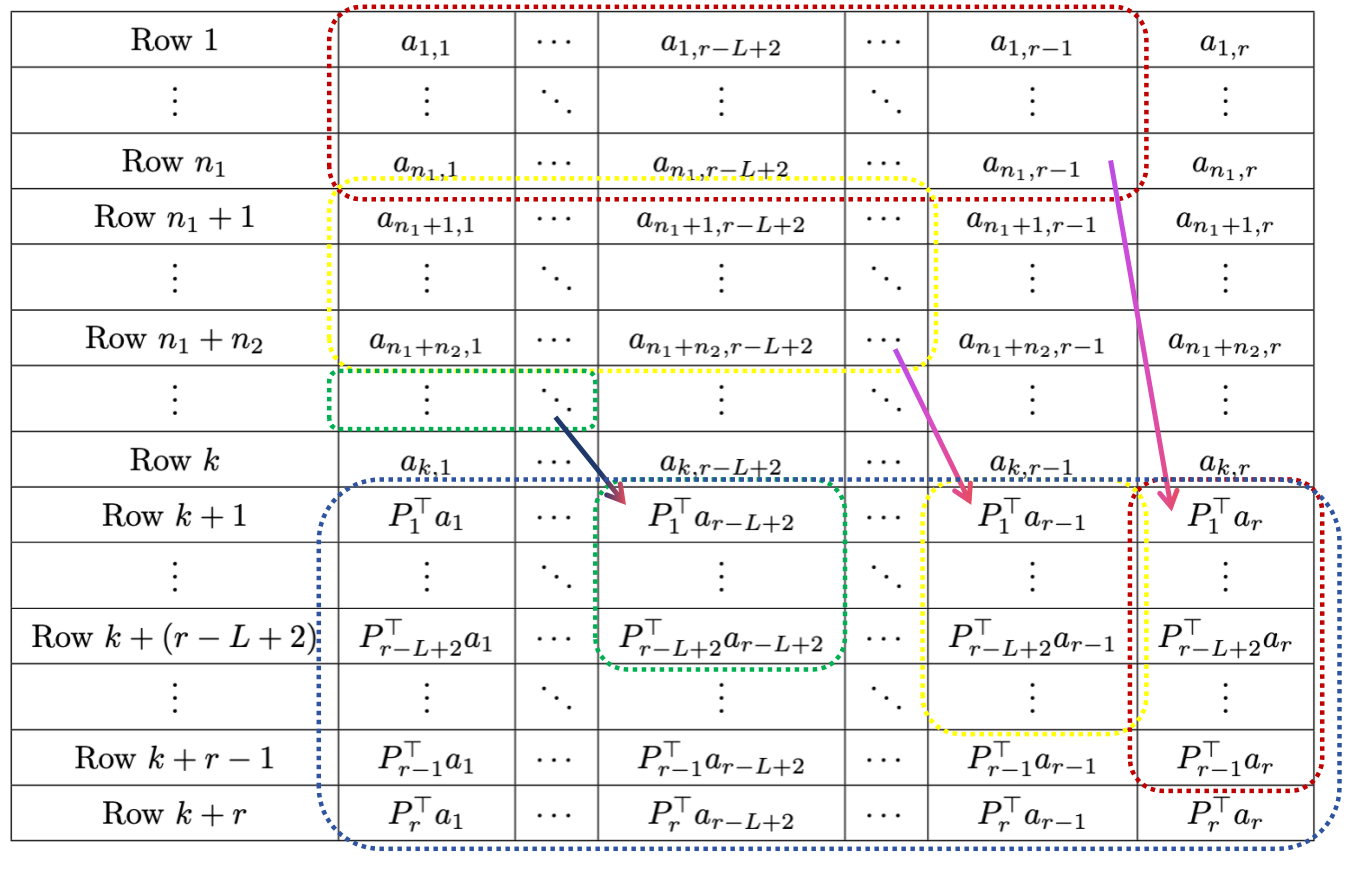}
	\caption{The structure of $\mathbf{G}_2$, where the blue part represents the region $\mathbf{R}$, the red part represents the region $\mathbf{M}_1$, and the yellow part represents the region $\mathbf{M}_2$.}
	\label{fig:1}
\end{figure}
%And for $i\in[r]$, $P(i,i)=R(i,i)=\mathbf{P}_i^T\mathbf{a}_{i}$.

%After the above three steps, we call the final array 
The obtained array $\mathbf{G}_3$ is the proposed conjugate-piggybacking code, denoted by $\mathcal{C}(n,k,L)$. 
%It can be verified that 
The example in Fig.~\ref{fig:0.1} corresponds to $\mathcal{C}(14,10,3)$.%, and in order to show the transformation in step 3 more clearly, we mark the increment of parity symbols located in $\mathbf{R}$ area in $\mathbf{G}_3$ caused by the transformation in step 3 as red.

\subsection{MDS Property}
\label{sec:2.2}
We next establish a sufficient condition under which the proposed codes are MDS.%this section, we prove that code $\mathcal{C}(n,k,L)$ has $(n,k)$ MDS property for any parameters $n,k,L$.

\begin{theorem}
    \label{th:1}
        If the field size satisfies $q> kr^2$, then the codes $\mathcal{C}(n,k,L)$ are MDS codes.
    \end{theorem}
    \begin{proof}
        It suffices to show that $\mathcal{C}(n,k,L)$ can 
        recover from any $r=n-k$ erased nodes.
        Without loss of generality, suppose that the $r$ erased nodes consist of 
        $t$ data nodes and $r-t$ parity nodes, where $0\leq t\leq r$. When $t=0$, all erased nodes are parity nodes, and the erased symbols can be recovered directly from the encoding procedure. We therefore consider $1\leq t\leq r$ and suppose that the 
        surviving $t$ parity nodes are $\{k+i_j\}_{j\in[t]}$, where $1\leq i_1<i_2<\cdots<i_t\leq r$.  
        Let $\{j_u\}_{u\in[t]}$ denote the $t$ erased data nodes, where $1\leq j_1<j_2<\cdots<j_{t}\leq k$. 
        
        We recover all erased data symbols in the $t$ erased data nodes through the following steps.

            \textbf{For all $v\in[t]$, recover data symbols $\{a_{j_u,i_v}\}_{u\in[t]}$.} First, we 
            download $t^2$ symbols $\{P(i_u,i_v)\}_{u,v\in[t]}$ to compute all 
            $\{R(i_v,i_v)\}_{v\in[t]}$ according to Eq.~\eqref{eq:2}.
            Since $R(i_u,i_1)=\mathbf{P}_{i_u}^T\mathbf{a}_{i_1}$ for $u\in[t]$, we can solve the erased data symbols
            $\{a_{j_u,i_1}\}_{u\in[t]}$ in column $i_1$ by the MDS property. Recall that, in the piggybacking step, $R(i_1,i_u)$ is a 
            linear combination of $\mathbf{a}_{i_1}$ and $\mathbf{P}_{i_1}^T\mathbf{a}_{i_u}$ for $u\in[2,t]$. Hence, we can obtain $\{\mathbf{P}_{i_1}^T\mathbf{a}_{i_u}\}_{u\in[2,t]}$.% by subtracting the part related             to $\mathbf{a}_{i_1}$ from $R(i_1,i_u)$.
            For $u\in[2,t]$, we have $R(i_u,i_2)=\mathbf{P}_{i_u}^T\mathbf{a}_{i_2}$, and $\mathbf{P}_{i_1}^T\mathbf{a}_{i_2}$ has already been obtained.% from the previous process. 
Therefore, we can solve all erased data symbols
            $\{a_{j_u,i_2}\}_{u\in[t]}$ in column $i_2$ by the MDS property. Similarly, since 
            $R(i_2,i_u)$ is a linear combination of $\mathbf{a}_{i_2}$ and $\mathbf{P}_{i_2}^T\mathbf{a}_{i_u}$ for $u\in[3,t]$ in the piggybacking step, 
            we can obtain $\{\mathbf{P}_{i_2}^T\mathbf{a}_{i_u}\}_{u\in[3,t]}$.
For $u\in[3,t]$, we can solve all $\{\mathbf{P}_{i_v}^T\mathbf{a}_{i_u}\}_{v\in[t]}$. 
            Therefore, for $u\in[3,t]$, %combining $\{\mathbf{P}_{i_v}^T\mathbf{a}_{i_u}\}_{v\in[t]}$, 
            we 
            can recover all $t$ erased data symbols $\{a_{j_v,i_u}\}_{v\in[t]}$ in column $i_u$ according to 
            the MDS property. 
            Thus, for every $v\in[t]$, all erased data symbols $\{a_{j_u,i_v}\}_{u\in[t]}$ in column $i_v$ are recovered.

            \textbf{For $r\geq v>i_t$, recover data symbols $\{a_{j_u,v}\}_{u\in[t]}$.} From Eq.~\eqref{eq:1}, 
            for $v\in[i_{t}+1,r]$ and $u\in[t]$, $P(i_u,v)=R(i_u,v)+\alpha\cdot R(v,i_u)$. Recall that, in the piggybacking step, $R(v,i_u)=\mathbf{P}_{v}^T\mathbf{a}_{i_u}$ 
            and $R(i_u,v)$ is a linear combination of $\mathbf{a}_{i_u}$ and $\mathbf{P}_{i_u}^T\mathbf{a}_{v}$. Since $\mathbf{a}_{i_u}$ 
            has been recovered, we can obtain $\mathbf{P}_{i_u}^T\mathbf{a}_{v}$. Hence, for $v\in[i_{t}+1,r]$, we have recovered $\{\mathbf{P}_{i_u}^T\mathbf{a}_{v}\}_{u\in[t]}$, which 
            means that all $t$ erased data symbols $\{a_{j_u,v}\}_{u\in[t]}$ in column $v$ can be solved by the MDS property. 

            \textbf{For $v\in U$, where $U:=[i_t]\setminus\{i_u\}_{u\in[t]}$, recover data symbols $\{a_{j_u,v}\}_{u\in[t]}$.}
            %According to the analyses of the first two steps, in order to repair all the failed data symbols, we only 
It remains to recover the erased data symbols in the column set $U$, namely $\{a_{j_u,v}\}_{u\in[t],v\in U}$.
            For $v\in U$, suppose $i_u<v<i_{u+1}$, where $u\in\{0,1,\cdots,t-1\}$ and $i_0=0$. From Eq.~\eqref{eq:1}, for  $j\in[t]$, we have
            \begin{equation}\nonumber
                P(i_j,v)=\left\{\begin{array}{l}
                    R(i_j,v)+\alpha\cdot R(v,i_j)=\\(\mathbf{P}_{i_j}^T\mathbf{a}_{v}+\overline{q}_{i_j,j}^T\mathbf{a}_{i_j})+\alpha\cdot\mathbf{P}_{v}^T\mathbf{a}_{i_j}, \forall 1\leq j\leq u\\ 
                    R(i_j,v)+R(v,i_j)=\\ \mathbf{P}_{i_j}^T\mathbf{a}_{v}+(\mathbf{P}_{v}^T\mathbf{a}_{i_j}+\overline{q}_{v,j}^T\mathbf{a}_v), \forall u+1\leq j\leq t
                \end{array}\right .\label{eq:3}
            \end{equation}
            where $\overline{q}_{i_j,j}^T\in\{0,q_{i_j,1}^T,q_{i_j,2}^T, \cdots ,q_{i_j,L-1}^T\}$ for $j\in[u]$ and $\overline{q}_{v,j}^T\in\{0,q_{v,1}^T,q_{v,2}^T,\cdots,q_{v,L-1}^T\}$ 
            for $j\in[u+1,t]$. Recall that all $\{\mathbf{a}_{i_j}\}_{j\in[t]}$ have been recovered in the first step. Therefore, for $j\in[t]$, according to Eq.~\eqref{eq:3}, we subtract 
            the part related to $\mathbf{a}_{i_j}$ from $P(i_j,v)$ and obtain the following $t$ symbols: $$\{\mathbf{P}_{i_1}^T\mathbf{a}_{v},\cdots,\mathbf{P}_{i_u}^T\mathbf{a}_{v},
            (\mathbf{P}_{i_{u+1}}^T+\overline{q}_{v,u+1}^T)\mathbf{a}_{v},\cdots,(\mathbf{P}_{i_{t}}^T+\overline{q}_{v,t}^T)\mathbf{a}_{v}\},$$ which are used to solve $\{a_{j_u,v}\}_{u\in[t]}$. For ease of notation, for $v\in U$, define the $t\times k$ matrix $\mathbf{M}_v$ as $$\mathbf{M}_v:=[\mathbf{P}_{i_1},\cdots,\mathbf{P}_{i_u},\mathbf{P}_{i_{u+1}}+\overline{q}_{v,u+1},\cdots,\mathbf{P}_{i_{t}}+\overline{q}_{v,t}]^T.$$
            It is equivalent to show that the submatrix, denoted by 
            $\mathbf{\overline{M}}_{v}$, composed of the $t$ rows and the columns $j_1,j_2,\cdots,j_t$ of $\mathbf{M}_v$ is invertible.

            %Next, we will show that $\mathbf{\overline{M}}_{v}$ is actually invertible. For the convenience of analysis, 
We define $\mathbf{M'}_v:=[\mathbf{P}_{i_1},\mathbf{P}_{i_2},\cdots,\mathbf{P}_{i_{t}}]^T$ 
            and let $\mathbf{\overline{M'}}_{v}$ be the submatrix composed of the $t$ rows and the columns $j_1,j_2,\cdots,j_t$ of $\mathbf{M'}_v$. By the MDS property, $\mathbf{\overline{M'}}_{v}$ is invertible. Note that for $j\in[u+1,t],\overline{q}_{v,j}^T\in\{0,q_{v,1}^T,q_{v,2}^T,\cdots,q_{v,L-1}^T\}$ 
            and $\mathbf{P}_i=(\alpha^i,\alpha^{2i},\cdots,\alpha^{ki})^T=\sum_{t\in[L]}\mathbf{q}_{i,t}$ for $i\in[r]$. Therefore, for $\delta\in[k]$, the $\delta$-th element of the $k\times1$ vector $\mathbf{P}_{i_{j}}+\overline{q}_{v,j}$  
            has the form $\alpha^{\delta\cdot i_j}+g_{v,\delta}(\alpha)$, where $g_{v,\delta}(\alpha)\in\{0,\alpha^{\delta\cdot v}\}$. Since $v<i_{u+1}\leq i_{j}$ for $j\in[u+1,v]$, we have $\deg(g_{v,\delta}(\alpha))\leq \delta\cdot v<\delta\cdot i_j$. 
            Here, for a polynomial $f(\alpha)$, $\deg(f(\alpha))$ denotes the degree of $f(\alpha)$. Thus, for $j\in[u+1,t],\delta\in[k]$, we have $\deg(\alpha^{\delta\cdot i_j}+g_{v,\delta}(\alpha))=\deg(\alpha^{\delta\cdot i_j})=\delta\cdot i_j$. 
            Then, for $i_j\in[u+1,t],\delta\in[k]$, we have $\deg(\alpha^{\delta\cdot i_j}+g_{v,\delta}(\alpha))=\deg(\alpha^{\delta\cdot i_j})=\delta\cdot i_j \leq kt \leq kr$ because $1\leq t \leq r$. 
            The determinant det$(\overline{M'}_v)$ is a polynomial in $\alpha$. Since each product term in the determinant polynomial contains at most $r$ factors of the form $\alpha^{\delta\cdot i_j}$ or $\alpha^{\delta\cdot i_j}+g_{v,\delta}(\alpha)$, 
            and the degree of each factor is less than $kr$, the degree of each product term in the determinant polynomial is at most $kr^2$. By [1, Theorem 1.2], %if the field size $q$ is larger than the degree of each $\alpha$ in the determinants multiplication, there exists at least one value for $\alpha$ such that the evaluation of the determinants multiplication is non-zero. Therefore, 
            if the field size $q$ is larger than $kr^2$, the determinant det$(\overline{M'}_v)$ is nonzero, i.e., the matrix $\overline{M'}_v$ is invertible. 

For any $v\in[r]$, all $t$ erased data symbols in column $v$ have been recovered, and the theorem is proved.%which means that the original information has been recovered, thus completing the proof of the MDS property of code $\mathcal{C}(n,k,L)$.
    \end{proof}

\section{Repair Process of Single-Node Erasure}
\label{sec:3}
In this section, we present the repair process for any single-node erasure of the proposed codes 
$\mathcal{C}(n, k, L)$. 
%We first consider the repair process of the data nodes.% and the parity nodes separately.

\subsection{Repair Process of Single Data Node}\label{sec:3.1}
Suppose that data node $f\in \Phi_i$ is erased, where $f\in [k]$ 
and $i \in [L-1]$. The repair procedure for node $f$ is as follows.

\begin{enumerate}
	\item[Step 1]  We download all $(k-1)i$ surviving data symbols in the last $i$ columns of the first $k$ rows 
    and $i$ parity symbols $\{P(r+1-u,r+1-u)=\mathbf{P}_{r+1-u}^T\mathbf{a}_{r+1-u}\}_{u\in[i]}$
    to recover the last $i$ erased symbols $\{a_{f,j}\}_{j\in[r+1-i,r]}$ in $f$ by the MDS property. Meanwhile, we can also compute $\{\mathbf{P}_u^T\mathbf{a}_{r+1-i}\}_{u \in [r-i]}$.
	
	\item[Step 2] For $v \in [r-i]$, we download the following two parity symbols:
	% \[
	% P(r+1-i,v)=\mathbf{P}_{r+1-i}^T\mathbf{a}_v + \mathbf{P}_v^T\mathbf{a}_{r+1-i} + \mathbf{q}^T_{v,i}\mathbf{a}_v, P(v,r+1-i)=\mathbf{P}_v^T\mathbf{a}_{r+1-i} + \mathbf{q}_{v,i}^T\mathbf{a}_v +\alpha\mathbf{P}_{r+1-i}^T\mathbf{a}_v
	% \]
        \begin{eqnarray*}\nonumber
            && P(r+1-i,v)\\&&=\mathbf{P}_{r+1-i}^T\mathbf{a}_v + \mathbf{P}_v^T\mathbf{a}_{r+1-i} + \mathbf{q}^T_{v,i}\mathbf{a}_v,\\ \nonumber
            && P(v,r+1-i)\\&&=\mathbf{P}_v^T\mathbf{a}_{r+1-i} + \mathbf{q}_{v,i}^T\mathbf{a}_v +\alpha\mathbf{P}_{r+1-i}^T\mathbf{a}_v  \nonumber
        \end{eqnarray*}
	Then, by Eq.~\eqref{eq:2}, we solve for $\mathbf{P}_{r+1-i}^T\mathbf{a}_v$ and $\mathbf{P}_v^T\mathbf{a}_{r+1-i} + \mathbf{q}^T_{v,i}\mathbf{a}_v$. %, i.e., 
%\begin{eqnarray}	 \nonumber
 %\begin{pmatrix}	\mathbf{P}_v^T\mathbf{a}_{r+1-i} + \mathbf{q}^T_{v,i}\mathbf{a}_v \\ \mathbf{P}_{r+1-i}^T\mathbf{a}_v	\end{pmatrix}=\begin{pmatrix}		1 & \alpha \\		1 &  1\\
%\end{pmatrix}^{-1}\cdot\\ \nonumber
 %\begin{pmatrix}
		%\mathbf{P}_v^T\mathbf{a}_{r+1-i} + \mathbf{q}_{v,i}^T\mathbf{a}_v +\alpha\mathbf{P}_{r+1-i}^T\mathbf{a}_v
		%\\\mathbf{P}_{r+1-i}^T\mathbf{a}_v + \mathbf{P}_v^T\mathbf{a}_{r+1-i} + \mathbf{q}^T_{v,i}\mathbf{a}_v 
%	\end{pmatrix}.\nonumber
% \end{eqnarray}
Then we download the $n_i - 1$ symbols $\{a_{s,v}\}_{s \in \Phi_i\backslash\{f\}}$ and recover the erased symbol $a_{f,v}$ as   
    % $$a_{f,v}=\alpha^{-fv}\cdot((\mathbf{P}_v^T\mathbf{a}_{r+1-i} + \mathbf{q}^T_{v,i}\mathbf{a}_v)-\mathbf{P}_v^T\mathbf{a}_{r+1-i}-\sum_{s\in\Phi_i\backslash\{f\}}\alpha^{sv}a_{s,v})$$.
    \begin{eqnarray}\nonumber
        a_{f,v}&=&\alpha^{-fv}\cdot((\mathbf{P}_v^T\mathbf{a}_{r+1-i} + \mathbf{q}^T_{v,i}\mathbf{a}_v)-\mathbf{P}_v^T\mathbf{a}_{r+1-i} \\ \nonumber
        & -&\sum_{s\in\Phi_i\backslash\{f\}}\alpha^{sv}a_{s,v}). \nonumber
    \end{eqnarray}
\end{enumerate}
The resulting repair bandwidth of node $f$ is $ki + 2(r-i)+(r-i)(n_i-1)=ki+(r-i)(n_i+1)$ symbols.

Suppose that the data node $f\in \Phi_L$ is erased. The repair process is as follows.

\begin{enumerate}
	\item[Step 1] We download all $(k-1)(L-1)$ surviving data symbols in the last $L-1$ columns of the first $k$ rows except row $f$, together with $L-1$ parity symbols $\{P(r+1-u,r+1-u)\}_{u \in[L-1]}$, to recover the last $L-1$ erased data symbols $\{a_{t, j}\}_{j\in[r+2-L,r]}$ in $f$ by the MDS property. 
	Meanwhile, we can also compute $\{\mathbf{P}_u^T\mathbf{a}_{r+1-\ell}\}_{u\in[r-L+1],\ell \in [L-1]}$.
	
%	\item We download $(L-1)(r - L+ 1)$ symbols in the last $L-1$ columns of row $k+s$ ($s\in[r-L+1]$) and download
%	$(L-1)(r - L+ 1)$ symbols in the last $L-1$ rows of column $k+ r+1 -s$ ($s\in[r-L+1]$). These $2(L-1)(r-L+1)$ symbols 
%	will form $(L-1)(r-L+1)$ binary equations systems of the same shape as Eq. \ref{eq:1}. By solving these systems of binary equations, we can obtain $\{\mathbf{P}_s^T\mathbf{a}_{r+1-j} 
%	+ \mathbf{q}^T_{s,j}\}_{s \in [r-L+1], j \in [L-1]}$. We download $r-L+1$ symbols $\{\mathbf{P}^T_s\mathbf{a}_s\}_{s \in [r-L+1]}$.
%	For $s = 1,2,\cdots,r-L+1$, we perform an invertible transformation by subtracting the summation of the $L-1$ symbols $\{\mathbf{P}_s^T\mathbf{a}_{r+1-j} 
%	+ \mathbf{q}^T_{s,j}\}_{j \in [L-1]}$ from the symbol $\mathbf{P}^T_s\mathbf{a}_s$. Since we have that
%	\begin{equation}\nonumber
%		\begin{split}\label{Eq:1}
%			\mathbf{P}_s^T\mathbf{a}_{s} - \sum_{\ell= 1}^{L - 1}(\mathbf{P}_s^T\mathbf{a}_{r - \ell + 1} + \mathbf{q}_{s,\ell}^T\mathbf{a}_{s})
%			= \mathbf{q}_{s,L}^T\mathbf{a}_{s} - \sum_{\ell =1}^{L-1}\mathbf{P}_s^T\mathbf{a}_{r-\ell + 1},
%		\end{split}
%	\end{equation}
%    we get $\{\mathbf{q}_{s,L}^T\mathbf{a}_{s}- \sum_{\ell =1}^{L-1}\mathbf{P}_s^T\mathbf{a}_{r-\ell + 1}\}_{s \in [r-L+1]}$. Then we download the symbols in the first $r - L + 1$ columns from all rows in $\Phi_L$ except row $t$, together with $\{\mathbf{P}_s^T\mathbf{a}_{r+1-\ell}\}_{s\in[r-L+1],\ell \in [L-1]}$, we can recover $\{a_{t,j}\}_{j \in [r+1-L]}$. The repair bandwidth of node $t$ is $k(L-1)+(r-L+1)n_L + 2(L-1)(r-L+1)$ symbols.
    
    \item[Step 2] For $v \in [r-L+1]$ and $u \in [r-L+2, r]$, we download the following two parity symbols:
    \begin{eqnarray}\nonumber
    	&& P(u,v)=\mathbf{P}_u^T\mathbf{a}_{v} + \mathbf{P}_v^T\mathbf{a}_u +\mathbf{q}^T_{v,r+1-u}\mathbf{a}_v, \\ \nonumber
     &&P(v,u)=\mathbf{P}_v^T\mathbf{a}_u + \mathbf{q}_{v,r+1-u}^T\mathbf{a}_v +\alpha\mathbf{P}_u^T\mathbf{a}_v, \nonumber
    \end{eqnarray}
    By Eq.~\eqref{eq:2}, these two symbols allow us to solve for $\mathbf{P}_u^T\mathbf{a}_v$ and $\mathbf{P}_v^T\mathbf{a}_u + \mathbf{q}_{v,r+1-u}^T\mathbf{a}_v$.
    % $$\begin{pmatrix}
    % 	\mathbf{P}_v^T\mathbf{a}_u + \mathbf{q}_{v,r+1-u}^T\mathbf{a}_v \\ \mathbf{P}_u^T\mathbf{a}_v
    % \end{pmatrix}=\begin{pmatrix}
    % 	1 & \alpha \\
    % 	1 &  1\\
    % \end{pmatrix}^{-1}\cdot\begin{pmatrix}
    % 	\mathbf{P}_v^T\mathbf{a}_u + \mathbf{q}_{v,r+1-u}^T\mathbf{a}_v +\alpha\mathbf{P}_u^T\mathbf{a}_v
    % 	\\\mathbf{P}_u^T\mathbf{a}_{v} + \mathbf{P}_v^T\mathbf{a}_u +\mathbf{q}^T_{v,r+1-u}\mathbf{a}_v
    % \end{pmatrix}.$$
    Next, we download $\mathbf{P}^T_v\mathbf{a}_v$ and subtract the sum of the $L-1$ symbols $\{\mathbf{P}_v^T\mathbf{a}_u + \mathbf{q}_{v,r+1-u}^T\mathbf{a}_v\}_{u \in [r-L+2, r]}$ from $\mathbf{P}^T_v\mathbf{a}_v$ to obtain
    \begin{eqnarray}\nonumber
    	\mathbf{P}^T_v\mathbf{a}_v - \sum_{u = r-L+2}^{r}(\mathbf{P}_v^T\mathbf{a}_u + \mathbf{q}_{v,r+1-u}^T\mathbf{a}_v)  \\= \mathbf{q}^T_{v,L}\mathbf{a}_v - \mathbf{P}_v^T(\sum_{u = r-L+2}^{r}\mathbf{a}_u). \nonumber
    \end{eqnarray}
    Then, we download $n_L - 1$ symbols $\{a_{s,v}\}_{s \in \Phi_L\backslash\{f\}}$ and recover the erased symbol $a_{f,v}$ as 
    % $$a_{f,v}=\alpha^{-fv}\cdot((\mathbf{P}^T_v\mathbf{a}_v - \sum_{u = r-L+2}^{r}(\mathbf{P}_v^T\mathbf{a}_u + \mathbf{q}_{v,r+1-u}^T\mathbf{a}_v))+\mathbf{P}_v^T(\sum_{u = r-L+2}^{r}\mathbf{a}_u)-(\sum_{s \in \Phi_L\backslash\{f\}}\alpha^{sv}a_{s,v})).$$
    \begin{eqnarray*}\nonumber
         &&a_{f,v}\\&&=\alpha^{-fv}\cdot((\mathbf{P}^T_v\mathbf{a}_v - \sum_{u = r-L+2}^{r}(\mathbf{P}_v^T\mathbf{a}_u + \mathbf{q}_{v,r+1-u}^T\mathbf{a}_v))\\
        &&+\mathbf{P}_v^T(\sum_{u = r-L+2}^{r}\mathbf{a}_u)-(\sum_{s \in \Phi_L\backslash\{f\}}\alpha^{sv}a_{s,v})). \nonumber
    \end{eqnarray*}
\end{enumerate}
Therefore, the repair bandwidth of node $f$ is $k(L-1)+(r-L+1)n_L + 2(L-1)(r-L+1)$ symbols.

Continuing the example in Fig.~\ref{fig:0.1}, we have $L = 3$ and $\Phi_1 = \{1,2,3,4\},\Phi_2 = \{5,6,7\},\Phi_3 = \{8,9,10\}$. The repair process of node 1 ($\in\Phi_1$) has been described in Section~\ref{sec:1.1}. We can also show that the repair bandwidth of each node in $\Phi_1$ is 25 symbols, and that of each node in $\Phi_2$ is 28 symbols.

Suppose that node 8 ($\in\Phi_3$) is erased. First, we download 20 symbols 
\begin{align*}
	&a_{1,3},a_{2,3},a_{3,3},a_{4,3},a_{5,3},a_{6,3},a_{7,3},a_{9,3},a_{10,3},\mathbf{P}_3^T\mathbf{a}_3,\\
	&a_{1,4},a_{2,4},a_{3,4},a_{4,4},a_{5,4},a_{6,4},a_{7,4},a_{9,4},a_{10,4},\mathbf{P}_4^T\mathbf{a}_4,
\end{align*}
to recover the erased symbols $a_{8,3}, a_{8,4}$ by the MDS property. For $v \in \{1,2\}$ and $u\in\{3,4\}$, we download the following two parity symbols:
\begin{eqnarray}\nonumber
	P(u,v)=\mathbf{P}_u^T\mathbf{a}_{v} + \mathbf{P}_v^T\mathbf{a}_u +\mathbf{q}^T_{v,r+1-u}\mathbf{a}_v,\\ P(v,u)=\mathbf{P}_v^T\mathbf{a}_u + \mathbf{q}_{v,r+1-u}^T\mathbf{a}_v +\alpha\mathbf{P}_u^T\mathbf{a}_v \nonumber
\end{eqnarray}
to solve for $\mathbf{P}_u^T\mathbf{a}_v$ and $\mathbf{P}_v^T\mathbf{a}_u + \mathbf{q}_{v,r+1-u}^T\mathbf{a}_v$.
Next, we download $\mathbf{P}^T_v\mathbf{a}_v$ and subtract the sum of the two symbols $\mathbf{P}_v^T\mathbf{a}_{3} 
+ \mathbf{q}^T_{v,2}\mathbf{a}_v$ and $\mathbf{P}_v^T\mathbf{a}_{4} 
+ \mathbf{q}^T_{v,1}\mathbf{a}_v$ from $\mathbf{P}^T_v\mathbf{a}_v$ to obtain
\begin{eqnarray}\nonumber
	\mathbf{P}^T_v\mathbf{a}_v - \Big(\mathbf{P}_v^T\mathbf{a}_{3} 
	+ \mathbf{q}^T_{v,2}\mathbf{a}_v+ \mathbf{P}_v^T\mathbf{a}_{4} 
	+ \mathbf{q}^T_{v,1}\mathbf{a}_v\Big)  \\ = \mathbf{q}^T_{v,3}\mathbf{a}_v - \mathbf{P}_v^T(\mathbf{a}_3 + \mathbf{a}_4).\nonumber
\end{eqnarray}
Then we download two symbols $a_{9,v},a_{10,v}$ to recover the erased symbol $a_{8,v}$ by $a_{8,v} = \alpha^{-8v}\cdot((\mathbf{q}^T_{v,3}\mathbf{a}_v - \mathbf{P}_v^T(\mathbf{a}_3 + \mathbf{a}_4)) + \mathbf{P}_v^T(\mathbf{a}_3 + \mathbf{a}_4 ) - (\alpha^{9v}a_{9,v} + \alpha^{10v}a_{10,v}))$. Therefore, the repair bandwidth of node 8 is $20 + 14 = 34$ symbols. The same argument shows that the repair bandwidth of each node in $\Phi_3$ is also 34 symbols.

\subsection{Repair Process of Single Parity Node}\label{sec:3.2}
Suppose that node $f$ is erased, where $f\in\{k+1,k+2,\cdots,k+r\}$. The repair process is as follows.

\begin{enumerate}
	\item[Step 1] We download all $k$ data symbols $\{a_{i,f-k}\}_{i\in[k]}$ in column $f-k$ to recover the erased parity symbol $\mathbf{P}^T_{f-k}\mathbf{a}_{f-k}(=P(f-k,f-k)=R(f-k,f-k))$ by the MDS property. Meanwhile, we can compute $\{\mathbf{P}_u^T\mathbf{a}_{f-k}\}_{u\in[r]\setminus\{f-k\}}$.
	
	\item[Step 2] We download $r-1$ parity symbols $\{P(u,f-k)\}_{u\in[r]\setminus\{f-k\}}$ in column $f-k$ from rows $k + 1$ to $k + r$, excluding row $f$. According to the piggybacking step in Section~\ref{sec:2.1}, $R(u,f-k)$ can be expressed as follows. 
 \begin{figure*}[ht] 
 	\centering
        \begin{equation}
             R(u,f-k)=\begin{cases}
            \mathbf{P}_u^T\mathbf{a}_{f-k} & \text{ if } 1\leq u<f-k, k+1\leq f\leq k+r-L+1; \\ 
            \mathbf{P}_u^T\mathbf{a}_{f-k}+\mathbf{q}_{u,k+r+1-f}^T\mathbf{a}_{u} & \text{ if } 1\leq u<f-k, k+r-L+1< f\leq k+r; \\ 
            \mathbf{P}_u^T\mathbf{a}_{f-k} & \text{ if } f-k<u\leq r. 
        \end{cases}
        \label{eq:111}
        \end{equation}
\end{figure*}
Then, by downloading all data symbols used to compute the piggyback functions in $\{R(u,f-k)\}_{u\in[f-k-1]}$ during the piggybacking step, together with the $r-1$ symbols $\{\mathbf{P}_u^T\mathbf{a}_{f-k}\}_{u\in[r]\setminus\{f-k\}}$, we can recover $\{R(u,f-k)\}_{u\in[r]\setminus\{f-k\}}$. 
When $f\in [k + 1, k+r-L+1]$, there is no piggyback function in $\{R(u,f-k)\}_{u\in[f-k-1]}$. When $f\in [k+r-L+2,k+r]$, the piggyback functions in $\{R(u,f-k)\}_{u\in[f-k-1]}$ are formed by all $\Phi_{k+r+1-f}(f-k-1)$ data symbols in region $\mathbf{M}_{k+r+1-f}$. Next, for $u\in[r]$, using the relationship among $P(u,f-k),R(u,f-k)$, and $R(f-k,u)$ in Eq.~\eqref{eq:1}, we can recover $\{R(f-k,u)\}_{u\in[r]\setminus\{f-k\}}$ as follows.
 \begin{figure*}[ht] 
 	\centering
        \begin{equation}    
    R(f-k,u)=\begin{cases}
    \alpha^{-1}\cdot(P(u,f-k)-R(u,f-k)) & \text{ if } 1\leq u\leq f-k-1; \\ 
    P(u,f-k)-R(u,f-k) & \text{ if } f-k+1\leq u\leq r.
    \end{cases}
       \label{eq:222}
    \end{equation} 
\end{figure*}

    \item[Step 3]
    Using $\{R(f-k,u)\}_{u\in[r]\setminus\{f-k\}}$ and $\{R(u,f-k)\}_{u\in[r]\setminus\{f-k\}}$, we recover all $r-1$ erased parity symbols $\{P(f-k,u)\}_{u\in[r]\setminus\{f-k\}}$ according to Eq.~\eqref{eq:1} as follows.
     \begin{figure*}[ht] 
 	\centering
        \begin{equation}    
P(f-k,u)=\begin{cases}
    R(f-k,u)+R(u,f-k) & \text{ if } 1\leq u\leq f-k-1; \\ 
    R(f-k,u)+\alpha R(u,f-k) & \text{ if } f-k+1\leq u\leq r.  
    \end{cases}
       \label{eq:333}
    \end{equation} 
\end{figure*}
\end{enumerate}
Therefore, when $f\in[k+r-L+2,k+r]$, the repair bandwidth of node $f$ is $k + r - 1 + n_{k+r+1-f}(f-k-1)$ symbols; when $f\in[k+1,k+r-L+1]$, the repair bandwidth of node $f$ is $k + r - 1$ symbols.

It is worth noting that the optimal repair bandwidth of an $(n,k,\ell=n-k)$ MDS array code is $\frac{(n-1)\ell}{n-k}=n-1=k+r-1$ symbols \cite{dimakis2010}. Thus, the repair bandwidth of each of the first $r-L+1$ parity nodes of $\mathcal{C}(n,k,L)$ is optimal.

%The repair process of node 11 has been described in Section \ref{sec:1.1}, and using the similar process we can repair node 12 and the repair bandwidth of node 12 is also optimal (13 symbols). 
Consider the code $\mathcal{C}(14,10,3)$ in Fig.~\ref{fig:0.1}. Suppose that node 13 is erased. First, we download 10 data symbols $\{a_{i,3}\}_{i\in[10]}$ to recover the erased parity symbol $\mathbf{P}_3^T\mathbf{a}_3$ by the MDS property, and then compute $\{\mathbf{P}_1^T\mathbf{a}_3, \mathbf{P}_2^T\mathbf{a}_3, \mathbf{P}_4^T\mathbf{a}_3\}$. Then, we download 6 symbols $\{a_{i,1}\}_{i\in[5,7]}, \{a_{s,2}\}_{s\in[5,7]}$ to recover the piggyback functions $\mathbf{q}^T_{1,2}\mathbf{a}_1(=\alpha^{5}a_{5,1}+\alpha^{6}a_{5,1}+\alpha^{7}a_{5,1})$ and $\mathbf{q}^T_{2,2}\mathbf{a}_2(=\alpha^{10}a_{5,2}+\alpha^{12}a_{5,2}+\alpha^{14}a_{5,2})$. Finally, we download 3 parity symbols,
$P(1,3)=\mathbf{P}_1^T\mathbf{a}_3 + \mathbf{q}_{1,2}^T\mathbf{a}_1 + \alpha\mathbf{P}_3^T\mathbf{a}_1, P(2,3)=\mathbf{P}_2^T\mathbf{a}_3 + \mathbf{q}_{2,2}^T\mathbf{a}_2 + \alpha\mathbf{P}_3^T\mathbf{a}_2, P(4,3)=\mathbf{P}_4^T\mathbf{a}_3 + \mathbf{P}_3^T\mathbf{a}_4 + \mathbf{q}_{3,1}^T\mathbf{a}_3$
and recover the 3 parity symbols as follows: 
% \begin{eqnarray}\nonumber
% 	& P(3,1)=& \alpha^{-1}(P(1,3) - \mathbf{P}_1^T\mathbf{a}_3 - \mathbf{q}_{1,2}^T\mathbf{a}_1) + \mathbf{P}_1^T\mathbf{a}_3 + \mathbf{q}_{1,2}^T\mathbf{a}_1, \\ \nonumber
% 	& P(3,2) =& \alpha^{-1}(P(2,3) - \mathbf{P}_2^T\mathbf{a}_3 - \mathbf{q}_{2,2}^T\mathbf{a}_2) + \mathbf{P}_2^T\mathbf{a}_3 + \mathbf{q}_{2,2}^T\mathbf{a}_2, \\ \nonumber
% 	& P(3,4) =& P(4,3) - \mathbf{P}_4^T\mathbf{a}_3 + \alpha\mathbf{P}_4^T\mathbf{a}_3.
% \end{eqnarray}
\begin{eqnarray}\nonumber
    & P(3,1) =& \alpha^{-1}(P(1,3) - \mathbf{P}_1^T\mathbf{a}_3 - \mathbf{q}_{1,2}^T\mathbf{a}_1) \\ \nonumber 
    & & + \mathbf{P}_1^T\mathbf{a}_3 + \mathbf{q}_{1,2}^T\mathbf{a}_1, \\ \nonumber
    & P(3,2) =& \alpha^{-1}(P(2,3) - \mathbf{P}_2^T\mathbf{a}_3 - \mathbf{q}_{2,2}^T\mathbf{a}_2) \\ \nonumber
    & & + \mathbf{P}_2^T\mathbf{a}_3 + \mathbf{q}_{2,2}^T\mathbf{a}_2, \\ \nonumber
    & P(3,4) =& P(4,3) - \mathbf{P}_4^T\mathbf{a}_3 + \alpha\mathbf{P}_4^T\mathbf{a}_3.
\end{eqnarray}
The repair bandwidth of node 13 is 19 symbols. %Using the similar method, the repair bandwidths of node 14 can be obtained as 25 symbols.

\section{Repair Bandwidth}
\label{sec:4}
In this section, we analyze the repair bandwidth of $\mathcal{C}(n,k,L)$. We define the {\em average repair bandwidth ratio} of data nodes, parity nodes, or all nodes as the ratio
of the corresponding average repair bandwidth to the total number
of data symbols.
\subsection{Average Repair Bandwidth Ratio of Data Nodes}
\label{sec:4.1}
\begin{lemma}
\label{lem:2}
	If $L$ is a factor of $k$, then the average repair bandwidth ratio of data nodes
	$\gamma_{sys}$ for codes $\mathcal{C}(n,k,L)$ is
	$$\gamma_{sys}=\frac{L}{2r} + \frac{1}{L} - \frac{3}{2Lr} + \frac{1}{L^2r} + \frac{-\frac{5}{2}L^2 + 3Lr +\frac{9}{2}L -3r -2}{kLr}.$$
\end{lemma}

\begin{proof}
	According to the repair method in Section~\ref{sec:3.1}, the average repair bandwidth ratio of data nodes $\gamma_{sys}$ is 
 \begin{eqnarray}\nonumber
     &\gamma_{sys} = \frac{(k(L-1)+(r-L+1)n_L + 2(L-1)(r-L+1))n_L}{k^2r} \\ \nonumber
     &+\frac{\sum_{i=1}^{L-1} (ki+(r-i)n_i + (r-i))n_i}{k^2r}.
 \end{eqnarray}
	If $L$ is a factor of $k$, then $n_i = \frac{k}{L}$ for $i\in[L]$. Hence, 
\begin{eqnarray*}\nonumber
 &&\gamma_{sys} \\&&= \frac{(k(L-1)+(r-L+1)\frac{k}{L} + 2(L-1)(r-L+1))\frac{k}{L}}{k^2r} \\ \nonumber
 &&+ \frac{\sum_{i=1}^{L-1} (ki+(r-i)\frac{k}{L} + (r-i))\frac{k}{L}}{k^2r},
\end{eqnarray*}
and further obtain 
	$$\gamma_{sys} = \frac{L}{2r} + \frac{1}{L} - \frac{3}{2Lr} + \frac{1}{L^2r} + \frac{-\frac{5}{2}L^2 + 3Lr +\frac{9}{2}L -3r -2}{kLr}.$$
\end{proof}

\begin{lemma}
\label{lem:3}
	If $L$ is a factor of $k$ and $k \to \infty$, then the repair bandwidth ratio of data nodes $\gamma_{sys}$ attains its minimum when $L = \lfloor\sqrt{2r - 1}\rfloor$ or $L = \lceil\sqrt{2r-1}\rceil $.
\end{lemma}

\begin{proof}
	If $L$ is a factor of $k$, then by Lemma~\ref{lem:2}, $\gamma_{sys} = \frac{L}{2r} + \frac{1}{L} - \frac{3}{2Lr} + \frac{1}{L^2r}$. Taking the partial derivative of $\gamma_{sys}$ with respect to $L$ gives $$\frac{\partial \gamma_{sys}}{\partial L} = \frac{1}{2r} - \frac{1}{L^2} + \frac{3}{2L^2r} - \frac{2}{L^3r} = 0,$$ i.e., $L^3 +  (3-2r)L-4 = 0$. Since $L$ is a positive integer, %according to the relevant conclusions in \cite{2021piggybacking}, 
 $\gamma_{sys}$ attains the approximate minimum when $L = \lfloor\sqrt{2r - 1}\rfloor$ or $L = \lceil\sqrt{2r-1}\rceil $.
\end{proof}

\subsection{Average Repair Bandwidth Ratio of Parity Nodes}

\begin{lemma}
\label{lem:4}
	If $L$ is a factor of $k$, then the average repair bandwidth ratio of parity nodes $\gamma_{par}$ is $$\gamma_{par} = \frac{k+r-1}{kr} + \frac{(L-1)(2r-L)}{2Lr^2}.$$
\end{lemma}

%\begin{proof}
%	According to the repair method in Section \ref{piggybackparity}, we need to
%	download $\sum_{i=1}^{r}k$ symbols in repairing each of the $r$ parity nodes in the first step.
%	In the second step, we need to download $\sum_{i=1}^{r}(r-1) + \sum_{i = 1}^{L-1} n_i(r-i)$
%	in repairing each of the $r$ parity nodes. Therefore, we can calculate that
%	\begin{equation}\nonumber
%		\gamma_{parity} = \frac{\sum_{i = 1}^{r}(k+r-1) + \sum_{i=1}^{L-1}n_i(r-i)}{kr^2}.
%	\end{equation}
%	Because $L$ is a factor of $k$, we have $n_i = \frac{k}{L}$ for $i = 1,2,\ldots,L-1$,
%	and further obtain that
%	$$\gamma_{parity} = \frac{k+r-1}{kr} + \frac{(L-1)(2r-L)}{2Lr^2}.$$
%\end{proof}

\begin{proof}
	According to the repair method in Section~\ref{sec:3.2}, when $f\in[k+r-L+2,k+r]$, the repair bandwidth of node $f$ is $k + r - 1 + n_{k+r+1-f}(f-k-1)$ symbols; when $f\in[k+1,k+r-L+1]$, the repair bandwidth of node $f$ is $k + r - 1$ symbols. 
    Therefore,
	\begin{eqnarray}\nonumber
		\gamma_{par} &=& \frac{\sum_{t = k+r-L+2}^{k+r}\Big[k + r - 1 + n_{k+r+1-f}(f-k-1)\Big] }{kr^2} \\ \nonumber
            &+& \frac{(r-L+1)(k+r-1)}{kr^2} \\ \nonumber
		&=& \frac{r(k+r-1) + \sum_{i=1}^{L-1}n_i(r-i)}{kr^2}. \nonumber
	\end{eqnarray}
	Since $L$ is a factor of $k$, we have $n_i = \frac{k}{L}$ for $i = 1,2,\cdots,L-1$,
	which gives
	$$\gamma_{par} = \frac{k+r-1}{kr} + \frac{(L-1)(2r-L)}{2Lr^2}.$$
\end{proof}

\section{Evaluation and Comparison}
\label{sec:5}
In this section, we evaluate the proposed conjugate-piggybacking codes $\mathcal{C}(n,k,L)$ and compare them with existing piggybacking-based constructions that are designed to reduce repair bandwidth. Our focus is on repair efficiency under small sub-packetization, since this regime is particularly relevant to high-rate distributed storage systems where both repair traffic and implementation overhead must be carefully controlled.

\subsection{Analytical Comparison}

Table~\ref{tab:1} summarizes the sub-packetization level and the theoretical lower bounds on repair bandwidth for representative piggybacking-based codes and for the proposed construction. In particular, we use $\mathcal{C}_0$ to denote the code in \cite{2021piggybacking}, and $\mathcal{C}_1$ to denote the code in \cite{2022Piggyback}. The table highlights the design trade-offs among sub-packetization, data-node repair bandwidth, and parity-node repair bandwidth.

\begin{table*}
	\centering
	\caption{Comparison of representative piggybacking-based codes and the proposed code.}
	\begin{tabular}{|c|c|c|c|}
		\hline
		 Codes&Sub-packetization&\makecell[c]{Average Repair Bandwidth Ratio of \\ Data Nodes \\(or all nodes in the first code in \cite{2023piggyback})} &\makecell{Average Repair Bandwidth Ratio of \\ Parity Nodes}\\
		\hline
		RSR-I \cite{2017Piggybacking}&$2h$ & $\geq\frac{r+1}{2r}$ & $\geq\frac{1}{r}+\frac{(r-1)(k+r-1)}{2kr}$ \\
		\hline
		RSR-II \cite{2017Piggybacking}&$(2r-3)h$ & $\geq\frac{r+1}{2r-3}$ & $\geq\frac{1}{r}+\frac{(r-1)(k+r-1)}{2kr}$ \\ 
		\hline
		REPB \cite{2018Repair}&$s+p$ & $\geq \frac{2}{\sqrt{r}+1}$ & 1\\
	    \hline
		OOP \cite{2019AnEfficient}&$r-1+\sqrt{r-1}$ & $\geq\frac{2\sqrt{r-1}+1}{2\sqrt{r-1}+r}$ & $\geq\frac{\sqrt{r-1}+1}{r} + \frac{(r-1)^2-\sqrt{(r-1)^3}}{kr}$ \\[3pt] 
		\hline
		Code $\mathcal{C}_0$ \cite{2021piggybacking}&$r$ &\makecell{$\geq \frac{L^*}{2r}+\frac{1}{L^*}-\frac{1}{2L^*r}$ \\ $(L^{*}=\sqrt{2r-1})$} & \makecell{$ \geq \frac{k+r}{kr}+\frac{2(r-L^*-1)^2}{k(4r-3-L^*)}$ \\  $+\frac{(kr-k-r)(L^*+1)}{kr^2}$\\$(L^{*}=\sqrt{2r-1})$}\\ 
		\hline
		Code $\mathcal{C}_1$ \cite{2022Piggyback}&$r$ & \makecell{$\geq\frac{L^{**}}{2r} + \frac{1}{L^{**}} - \frac{3}{2L^{**}r} + \frac{1}{(L^{**})^2r}$ \\ $(L^{**}=\sqrt{2r-1})$} &\makecell{$\geq\frac{(r-L^{**})(r-1)}{kr^2}+\frac{L^{**}}{r }+\frac{r-L^{**}}{r^2}$ \\ $+\frac{2(r-L^{**})^2(r-1)^2}{[(L^{**}-2)(L^{**}-1)+2r-4]kr^2}$\\$(L^{**}=\sqrt{2r-1})$} \\
		\hline
        The first code in \cite{2023piggyback}& $m<r$ &\makecell{$=\frac{(k+r)(m^2-m(L+1)+\frac{(L+1)(2L+1)}{6})}{Lmk(r-1)}$\\$+\frac{L+1}{2m}+\frac{m-L}{mk}$} &
        \\
		\hline
		Proposed code $\mathcal{C}(n,k,L)$ & $r$ &\makecell{ $\geq\frac{L}{2r} + \frac{1}{L} - \frac{3}{2Lr} + \frac{1}{L^2r}$\\$ + \frac{- L^3 +(2+r)L^2 -2L -r+1}{kLr}$ \\ ($L=\sqrt{2r-1}$)}& \makecell{$\geq\frac{k+r-1}{kr} + \frac{(L-1)(2r-L)}{2Lr^2}$\\ $(L=\sqrt{2r-1})$ } \\ 
		\hline
	\end{tabular}
    \label{tab:1}
\end{table*}

Existing results already establish a partial ordering among several baselines. In particular, \cite{2019AnEfficient} proved that when $r>5$, the OOP code has strictly lower repair bandwidth than REPB, RSR-I, and RSR-II. The numerical results in \cite{2021piggybacking} further show that when $r>10$ and $k$ is large, the repair bandwidth of $\mathcal{C}_0$ is strictly lower than that of OOP. It was then shown in \cite{2022Piggyback} that the lower bound on the average repair bandwidth ratio of data nodes of $\mathcal{C}_1$ is strictly lower than that of $\mathcal{C}_0$ for $r>0$. In addition, the numerical results in \cite{2023piggyback} indicate that the first code in \cite{2023piggyback} achieves the lowest average repair bandwidth ratio of all nodes among piggybacking codes with sub-packetization smaller than $r$ under high-code-rate settings.

Motivated by these results, we focus on the high-rate regime in which $r\ll k$ and the sub-packetization level remains on the order of $O(r)$. This is the regime in which piggybacking-based designs are most attractive for storage deployments, because it balances repair efficiency against implementation complexity. In the following theoretical comparison, we therefore compare the proposed construction mainly against OOP \cite{2019AnEfficient}, code $\mathcal{C}_1$ in \cite{2022Piggyback}, and the first code in \cite{2023piggyback}, which represent strong baselines under small-sub-packetization conditions. Moreover, when $\sqrt{2r-1}$ is an integer, Lemma~\ref{lem:4} shows that the proposed code achieves its minimum repair bandwidth at $L=\sqrt{2r-1}$.

\begin{theorem}{(Comparison with OOP code)}
	\label{th:5}
	When $5<r\ll k$, the average repair bandwidth ratios of both data nodes and parity nodes of the proposed code are strictly lower than those of the OOP code.
\end{theorem}
\begin{proof}
We first consider the average repair bandwidth ratio of data nodes. %of code $\mathcal{C}$ is strictly lower than that of OOP code.

From Table~\ref{tab:1}, the lower bound on the average repair bandwidth ratio of data nodes of the OOP code is $\gamma_{1}=\frac{2\sqrt{r-1}+1}{2\sqrt{r-1}+r}$. When $r\ll k$, the lower bound on the average repair bandwidth ratio of data nodes of the proposed code is $\gamma_{2}=\frac{\sqrt{2r-1}}{2r}+\frac{1}{\sqrt{2r-1}}-\frac{3}{2r\sqrt{2r-1}}+\frac{1}{r(2r-1)}$. Let $f(r):=\gamma_1-\gamma_2$. Then,
\begin{figure*}[ht]
\centering
\begin{eqnarray}
    &&f(r):=\gamma_1-\gamma_2\nonumber\\
    &&=\frac{2\sqrt{2r-1}+1}{2\sqrt{2r-1}+r}-(\frac{\sqrt{2r-1}}{2r}+\frac{1}{\sqrt{2r-1}}-\frac{3}{2r\sqrt{2r-1}}+\frac{1}{r(2r-1)})\nonumber\\
    &&=\frac{(2\sqrt{r-1}\sqrt{2r-1}(r(\sqrt{2r-1}-\sqrt{r-1})-2r+2))+(2r(r-1)-2\sqrt{r-1})}{r(2\sqrt{r-1}+r)(2r-1)}.
\end{eqnarray}
\end{figure*}
Let $g_1(r)=2\sqrt{r-1}\sqrt{2r-1}(r(\sqrt{2r-1}-\sqrt{r-1})-2r+2)$ and $g_2(r)=2r(r-1)-2\sqrt{r-1}$. We have $f(r)=\frac{g_1(r)+g_2(r)}{r(2\sqrt{r-1}+r)(2r-1)}$, and we only need to show that $g_1(r)+g_2(r)>0$.

Note that when $r\geq6$, we have $r\sqrt{r-1}-1\geq3\sqrt{2r-1}$. Thus $g_2(r)=2\sqrt{r-1}(r\sqrt{r-1}-1)\geq2\sqrt{r-1}\cdot3\sqrt{2r-1}$. Then, we have
\begin{eqnarray}\nonumber
    &&g_1(r)+g_2(r)\geq g_1(r)+2\sqrt{r-1}\cdot3\sqrt{2r-1}\\ \nonumber
    &&=2\sqrt{r-1}\sqrt{2r-1}(5+r(\sqrt{2r-1}-\sqrt{r-1}-2)).\nonumber
\end{eqnarray}
Let $h(r):=\sqrt{2r-1}-\sqrt{r-1}-2$. Since $h(23)\approx0.02>0$ and $h(r)$ is a monotone increasing function of $r$ when $r\geq23$, we have $h(r)\geq h(23)>0$, which means that $g_1(r)+g_2(r)\geq2\sqrt{r-1}\sqrt{2r-1}(5+r\cdot h(r))>0$ for $r\geq23$. On the other hand, for $r\in\{6,7,\cdots,22\}$, we can calculate $f(r)>0$. %Combined with the above analysis, we can know that for any $5<r\ll k$, $f(r)>0$, 
Therefore, the average repair bandwidth ratio of data nodes of the proposed code is strictly lower than that of the OOP code.

%Next, we show that the average repair bandwidth ratio of parity nodes of code $\mathcal{C}$ is strictly lower than that of OOP code. 

From Table~\ref{tab:1}, when $r\ll k$, the lower bound on the average repair bandwidth ratio of parity nodes of the OOP code is $\gamma_{3}=\frac{\sqrt{r-1}+1}{r}$, and the lower bound on the average repair bandwidth ratio of parity nodes of the proposed code is $\gamma_4=\frac{1}{r}+\frac{(L-1)(2r-L)}{2Lr^2}$, where $L=\sqrt{2r-1}$. Let 
\begin{eqnarray}
    &&\overline{f}(r)=\gamma_3-\gamma_4=\frac{\sqrt{r-1}+1}{r}-(\frac{1}{r}+\frac{(L-1)(2r-L)}{2Lr^2})\nonumber\\
    &&=\frac{Lr\sqrt{r-1}-(L-1)(2r-L)}{Lr^2}+\frac{(L-1)(2r-L)}{2Lr^2}.\nonumber
\end{eqnarray}
Note that when $r>5$, $r(2-\sqrt{r-1})<0<L$, i.e., $2r-L<r\sqrt{r-1}$, and hence $$(L-1)(2r-L)<L(2r-L)<Lr\sqrt{r-1}.$$
Therefore, when $r>5$, $\overline{f}(r)>\frac{(L-1)(2r-L)}{2Lr^2}>0$. Thus, the average repair bandwidth ratio of parity nodes of the proposed code is strictly lower than that of the OOP code.
\end{proof}

\begin{theorem}{(Comparison with code $\mathcal{C}_1$)}
	\label{th:6}
	When $1<r\ll k$, the average repair bandwidth ratio of data nodes of the proposed code is the same as that of code $\mathcal{C}_1$, and the average repair bandwidth ratio of parity nodes of the proposed code is strictly lower than that of code $\mathcal{C}_1$.
\end{theorem}
\begin{proof}
From Table~\ref{tab:1}, when $r\ll k$, the lower bounds on the average repair bandwidth ratios of data nodes of the proposed code and code $\mathcal{C}_1$ are the same. The lower bound on the average repair bandwidth ratio of parity nodes of the proposed code is $\gamma_5=\frac{1}{r}+\frac{(L-1)(2r-L)}{2Lr^2}$, where $L=\sqrt{2r-1}$. The lower bound on the average repair bandwidth ratio of parity nodes of code $\mathcal{C}_1$ is $\gamma_6=\frac{L^{**}}{r}+\frac{r-L^{**}}{r^2}$, where $L^{**}=L=\sqrt{2r-1}$. Let 
\begin{eqnarray}
    &&\overline{g}(r)=\gamma_6-\gamma_5\nonumber\\
    &&=(\frac{L^{**}}{r}+\frac{r-L^{**}}{r^2})-(\frac{1}{r}+\frac{(L-1)(2r-L)}{2Lr^2})\nonumber\\
    &&=\frac{r-L}{r^2}+\frac{(L-1)(2r-L)}{2Lr^2}.\nonumber
\end{eqnarray}
When $1<r$, we have $L=\sqrt{2r-1}<r$ and $\overline{g}(r)>\frac{(L-1)(2r-L)}{2Lr^2}>0$. Therefore, the lower bound on the average repair bandwidth ratio of parity nodes of the proposed code is strictly lower than that of code $\mathcal{C}_1$.
\end{proof}

\begin{theorem}{(Comparison with the first code in \cite{2023piggyback})}
	\label{th:7}
	When $r\ll k$ and $r\in\{5,6,7\}$, the average repair bandwidth ratio of all nodes of the proposed code is strictly lower than that of the first code in \cite{2023piggyback}.
\end{theorem}
\begin{proof}
From Table~\ref{tab:1}, when $r\ll k$, the lower bounds on the average repair bandwidth ratios of all nodes of the proposed code and the first code in \cite{2023piggyback} are $\gamma_7=\frac{\sqrt{2r-1}}{2r}+\frac{1}{\sqrt{2r-1}}-\frac{3}{2r\sqrt{2r-1}}+\frac{1}{r(2r-1)}=\frac{2(r-1)}{r\sqrt{2r-1}}+\frac{1}{r(2r-1)}$ and $\gamma_8(m,L)=\frac{L+1}{2m}+\frac{m^2-m(L+1)+\frac{(L+1)(2L+1)}{6}}{Lm(r-1)}$, respectively, where $1\leq L<m<r$ \cite{2023piggyback}. 
It is sufficient to show that, when $r\in\{5,6,7\}$, $$\underset{1\leq L<m<r}{\min}\gamma_8(m,L)>\gamma_7=\frac{2(r-1)}{r\sqrt{2r-1}}+\frac{1}{r(2r-1)}.$$

When $L=1$, we have
\begin{eqnarray} \nonumber
    \gamma_8(m,L=1)&=&\frac{1}{m}+\frac{(m-1)^2}{m(r-1)}\\ \nonumber
    &=&(\frac{r}{r-1}\cdot\frac{1}{m}+\frac{m}{r-1})-\frac{2}{r-1}\\ \nonumber
    &\geq&\frac{2\sqrt{r}}{r-1}-\frac{2}{r-1}=\frac{2}{\sqrt{r}+1},\nonumber
\end{eqnarray}
where the last inequality follows from $\frac{r}{r-1}\cdot\frac{1}{m}+\frac{m}{r-1}\geq2\sqrt{\frac{r}{r-1}\cdot\frac{1}{m}\cdot \frac{m}{r-1}}=\frac{2\sqrt{r}}{r-1}$.

Recall from Table~\ref{tab:1} that the lower bound on the average repair bandwidth ratio of data nodes of REPB is $\frac{2}{\sqrt{r}+1}$, which is strictly larger than that of OOP codes when $r>5$ \cite{2019AnEfficient}. We have shown in Theorem~\ref{th:5} that, when $r>5$, the lower bound on the average repair bandwidth ratio of data nodes of the proposed code is lower than that of the OOP code. Therefore, when $r>5$ and $L=1$, the proposed code has lower repair bandwidth than $\gamma_8(m,L=1)$. When $r=5$, direct calculation shows that $\frac{2}{\sqrt{r}+1}-\gamma_7\approx0.06>0$. Hence, $\underset{1=L<m<r}{\min}\gamma_8(m,L)>\gamma_7$ for $r\geq5$.

We next consider the case $2\leq L\leq m-1$. Let $Q$ denote the triangular region $Q:=\{(m,L)|2\leq L\leq m-1,m\leq r-1\}=\bigtriangleup ABC$, as shown in Fig.~\ref{fig:3}. %Combined with the above analysis, and note that $L,m$ are positive integers, 
It is equivalent to show that $$\underset{(m,L)\in Q}{\min}\gamma_8(m,L)>\gamma_7=\frac{2(r-1)}{r\sqrt{2r-1}}+\frac{1}{r(2r-1)}.$$

\begin{figure}[htpb]
	\centering
	\includegraphics[width=0.9\linewidth]{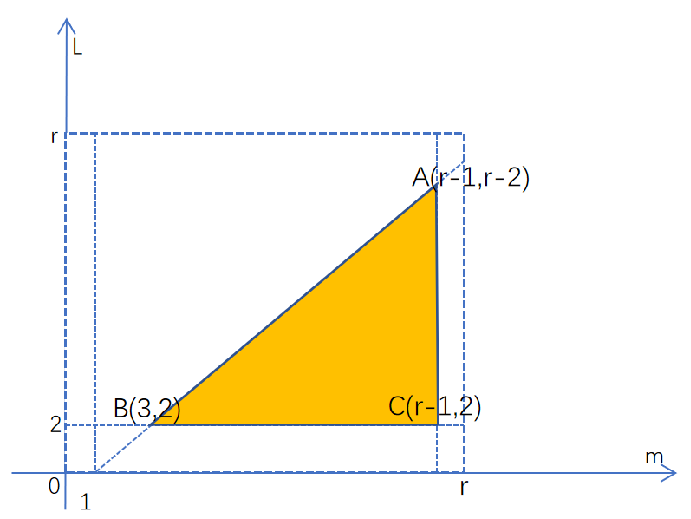}
	\caption{The triangle area $Q$, i.e., the orange part ($\bigtriangleup ABC$).}
	\label{fig:3}
\end{figure}

Since $Q$ is a bounded closed domain and $\gamma_8(m,L)$ is a continuous function of two variables on $Q$, the minimum value is attained. We next show by contradiction that the minimum of $\gamma_8(m,L)$ cannot be attained at an interior point of $Q$.

Suppose that $\gamma_8(m,L)$ attains the minimum at an interior point $(m',L')$ of $Q$. Then 
\begin{eqnarray}\nonumber
    &&\frac{\partial\gamma_8(m',L')}{\partial m}=0\Rightarrow \frac{1}{L'(r-1)}-\\ \nonumber
    &&(\frac{L'+1}{2}+\frac{1}{L'(r-1)}\cdot\frac{(L'+1)(2L'+1)}{6})\cdot\frac{1}{m'^2}=0 \nonumber\\
    &&\Rightarrow 6m'^2=(3r-1)L'^2+3rL' +1,\label{eq:10} \nonumber
\end{eqnarray}
and 
\begin{eqnarray} \nonumber
    &&\frac{\partial\gamma_8(m',L')}{\partial L}=0\\ \nonumber\Rightarrow
    &&\frac{1}{2m'}+\frac{1}{3m'(r-1)}-\frac{(m'^2-m'+\frac{1}{6})}{m'(r-1)L'^2}=0\nonumber\\ \nonumber
    &&\Rightarrow (3r-1)L'^2=6m'^2-6m'+1.\label{eq:11} \nonumber
\end{eqnarray}
Then, we have $6m'=3rL'+2$. Since $L'\geq2$ and $m'<r$, we obtain $3rL'+2\geq6r+2>6r>6m'$, which contradicts $6m'=3rL'+2$. Therefore, the minimum value of $\gamma_8(m,L)$ cannot be attained at an interior point of $Q$, i.e.,
$$\underset{(m,L)\in Q}{\min}\gamma_8(m,L)=\underset{(m,L)\in AB\bigcup BC\bigcup AC}{\min}\gamma_8(m,L).$$

When $(m,L)\in AB$, i.e., $m=L+1$, we have 
\begin{eqnarray}
    &&\gamma_8(m=L+1,L)=\frac{1}{2}+\frac{(\frac{L+1}{2})^2+\frac{1}{12}(L^2-1)}{Lm(r-1)}\nonumber\\
    &&>\frac{1}{2}+\frac{(\frac{L+1}{2})^2}{Lm(r-1)}=\frac{1}{2}+\frac{L+1}{4L(r-1)}>\frac{1}{2}+\frac{1}{4(r-1)}.\nonumber
\end{eqnarray}
Note that $\frac{1}{2}>\frac{2(r-1)}{r\sqrt{2r-1}}$ and $\frac{1}{4(r-1)}>\frac{1}{(2r-1)r}$ for $r\geq 7$. Thus,
\begin{align*}
\gamma_8(m=L+1,L)>&\frac{1}{2}+\frac{1}{4(r-1)}\\>&\frac{2(r-1)}{r\sqrt{2r-1}}+\frac{1}{(2r-1)r}=\gamma_7
\end{align*}
for $r\geq 7$.
When $r=5,6$, direct calculation shows that $(\frac{1}{2}+\frac{1}{4(r-1)})-(\frac{2(r-1)}{r\sqrt{2r-1}}+\frac{1}{(2r-1)r})>0$, and hence $\gamma_8(m=L+1,L)>\gamma_7$ also holds for $r=5,6$. Therefore, $\underset{(m,L)\in AB}{\min}\gamma_8(m,L)>\gamma_7$ for $r\geq5$.

When $(m,L)\in BC$, i.e., $L=2$, we have
\begin{eqnarray}\nonumber
    \gamma_8(m,L=2)&=&\frac{6r-1}{4(r-1)}\cdot\frac{1}{m}+\frac{m}{2(r-1)}-\frac{3}{2(r-1)}\\ \nonumber
    &\geq&\frac{\sqrt{3r-\frac{1}{2}}}{r-1}-\frac{3}{2(r-1)}, \nonumber
\end{eqnarray}
where the last inequality follows from $\frac{6r-1}{4(r-1)}\cdot\frac{1}{m}+\frac{m}{2(r-1)}\geq2\sqrt{\frac{6r-1}{4(r-1)}\cdot\frac{1}{m}\cdot \frac{m}{2(r-1)}}=\frac{\sqrt{3r-\frac{1}{2}}}{r-1}$.
Direct calculation shows that $\frac{\sqrt{3r-\frac{1}{2}}}{r-1}-\frac{3}{2(r-1)}>\gamma_7$ for $r\in\{2,3,\cdots,20000\}$. Thus, $\underset{(m,L)\in BC}{\min}\gamma_8(m,L)>\gamma_7$ for $r\in[2,20000]$.

When $(m,L)\in AC$, i.e., $m=r-1$, according to \cite[Lemma 4]{2023piggyback}, $\gamma_8(m,L)$ attains its minimum when $L=\sqrt{\frac{6m^2-6m+1}{3r-1}}$. Then $\underset{(m,L)\in AC}{\min}\gamma_8(m,L)=\gamma_8(m=r-1,L=\sqrt{\frac{6m^2-6m+1}{3r-1}})$.
Direct calculation shows that $\gamma_8(m=r-1,L=\sqrt{\frac{6m^2-6m+1}{3r-1}})-\gamma_7>0$ when $r=5,6,7$. Therefore, $\underset{(m,L)\in AC}{\min}\gamma_8(m,L)>\gamma_7$ for $r=5,6,7$.

Therefore, when $r\in\{5,6,7\}$ and $r\ll k$, the average repair bandwidth ratio of all nodes of the proposed code is strictly lower than that of the first code in \cite{2023piggyback}.
\end{proof}

%Theorem \ref{th:5}, Theorem \ref{th:6} and Theorem \ref{th:7} are all proved under the assumption of $r\ll k$, but in fact, $k$ will not get $\infty$. 
% In Fig. \ref{fig:4}, we show the comparison of the average repair bandwidth ratio of all nodes of code $\mathcal{C}$ and OOP code, code $\mathcal{C}_2$ and code $\mathcal{C}_1(n,k,m,L)$ when $r\in\{4,5,\cdots,20\}$ and $k=3r$ (see Fig. \ref{fig:4.(a)}) or $k=36$ (see Fig. \ref{fig:4.(b)}).

\subsection{Finite-Field Feasibility and Repair-Bandwidth Evaluation}

Table~\ref{tab:codes_and_fields} compares the proposed code with several related constructions, including HTEC \cite{2018hetc}, BPD \cite{2023wk1}, ET-RS codes \cite{2023tran}, piggybacking codes \cite{2017Piggybacking}, and piggybacking+ codes \cite{2023gcplus}, from the perspective of two implementation-relevant parameters: field size and sub-packetization level. These two parameters are important in practice because they directly affect the feasibility and complexity of code deployment in distributed storage systems.

From Table~\ref{tab:codes_and_fields}, we observe that the proposed code requires a smaller field size than HTEC, BPD, and ET-RS codes, while incurring only a modest field-size increase relative to conventional piggybacking codes \cite{2017Piggybacking} and piggybacking+ codes \cite{2023gcplus}. At the same time, our construction keeps the sub-packetization level at $\ell=r$, whereas piggybacking+ codes require a larger sub-packetization level of the form $\ell=sr$ with $2\leq s\leq r$. Therefore, the proposed design offers a more favorable trade-off between field-size overhead and sub-packetization among existing related constructions.

\begin{table}[htpb] 
\centering
\caption{Comparison of field size and sub-packetization level for the proposed code and related constructions.}
  \resizebox{\linewidth}{!}{
		\begin{tabular}{|c|c|c|}
			\hline
    Codes & Field Size $q$ & Sub-packetization \\
    \hline
    HTEC \cite{2018hetc} & $\geq\binom{n}{k}(n-k)\alpha$ & $2 \leq \ell\leq r^{\lceil\frac{k}{r}\rceil}$\\
    \hline
    BPD \cite{2023wk1} & $\geq\binom{n-1}{k-1}+2$& $2\leq \ell\leq r$\\
    \hline
    ET-RS \cite{2023tran} & $\geq2\Big(\binom{n - 1}{k - 1} - \binom{\lceil \frac{n}{\alpha*} \rceil - 1}{\lceil \frac{k}{\alpha*} \rceil - 1}\Big)$& $2\leq \ell\leq r^{\lfloor\frac{n}{r}\rfloor}$\\
    \hline
    Piggybacking codes \cite{2017Piggybacking} & $\geq n$ & Flexible\\
    \hline
    Piggybacking+ codes \cite{2023gcplus} & $\geq n$ & $\ell=sr$, where $2\leq s\leq r$\\
    \hline
    \textbf{Proposed codes} & $\geq kr^2$ & $\ell=r$\\
            \hline
		\end{tabular}
    \label{tab:codes_and_fields}}
\end{table}

Fig.~\ref{fig:4} further compares the average repair bandwidth ratio of all nodes for the proposed code and several efficient piggybacking-based schemes, including RSR-I \cite{2017Piggybacking}, OOP \cite{2019AnEfficient}, code $\mathcal{C}_1$ \cite{2022Piggyback}, and the code in \cite{2023gcpig}, denoted here as $\mathcal{C}_2$. The evaluation is conducted under the high-rate setting $r=\ell=4$, $4\leq k\leq 52$, and field $\mathbb{F}_{2^8}$. For the proposed construction, we set $L=3$.

For all the parameter settings shown in Fig.~\ref{fig:4}, the proposed code is MDS, as ensured by Theorem~\ref{th:1} and further verified by computer search over $\mathbb{F}_{2^8}$. In contrast, some alternative constructions, such as HTEC, ET-RS, and BPD, are not MDS over $\mathbb{F}_{2^8}$ for part of this parameter range. For example, when $k=25$ and $r=4$, the lower bound on the field size of BPD is $q\geq 20477$ according to Table~\ref{tab:codes_and_fields}, which is far beyond $\mathbb{F}_{2^8}$.

The results in Fig.~\ref{fig:4} show that, under all evaluated parameter settings, the proposed code achieves the lowest average repair bandwidth ratio among the compared schemes. Combined with the field-size and sub-packetization comparison in Table~\ref{tab:codes_and_fields}, this indicates that the proposed conjugate-piggybacking design provides a favorable operating point for high-rate distributed storage: it reduces repair traffic more effectively than existing related piggybacking constructions, while retaining small sub-packetization and MDS feasibility over a practically manageable field size.

\begin{figure}[htpb]
	\centering
	\includegraphics[width=0.86\linewidth]{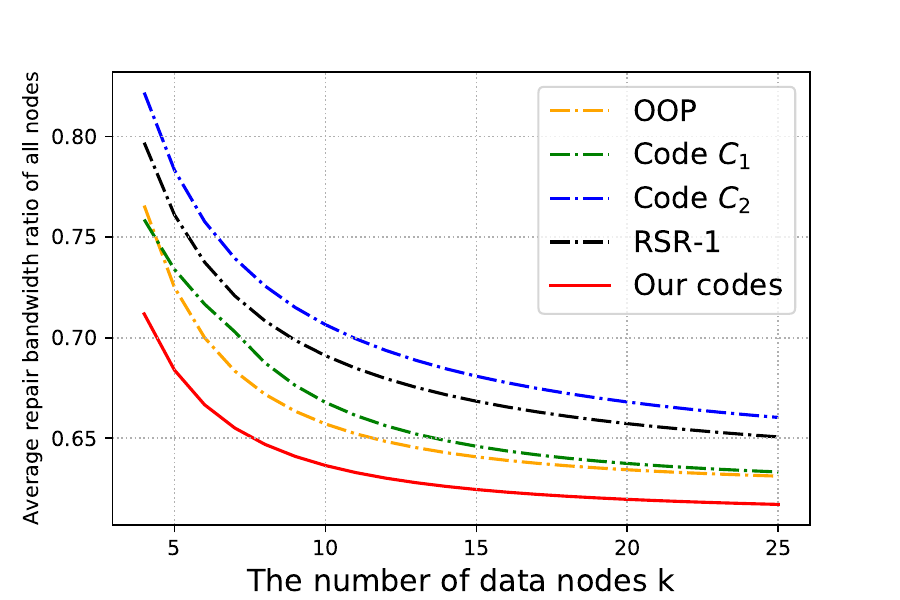}
	\caption{Average repair bandwidth ratio of all nodes for the proposed code, OOP code, code $\mathcal{C}_1$, code $\mathcal{C}_2$, and RSR-I code.}
	\label{fig:4}
\end{figure}

\subsection{Repair-Traffic Simulation}

To further connect the repair-bandwidth analysis with storage-oriented repair traffic, we also perform a single-node repair-traffic simulation. The simulation enumerates all possible single-node failures and applies the repair procedures described in Section~\ref{sec:3}. Since the failure distribution is uniform over the $n$ storage nodes, this event enumeration is equivalent to a Monte Carlo simulation with vanishing sampling error. The reported values are normalized by the number of data symbols $kr$ in the stripe, so a conventional systematic RS repair has normalized repair traffic equal to $1$.

Table~\ref{tab:repair_simulation} reports the simulated repair traffic for representative high-rate parameters. For $r=4$, we use $L=3$, consistent with Fig.~\ref{fig:4}. For the other cases, $L$ is chosen to minimize the simulated all-node repair traffic among $2\leq L\leq r$. The results show that the proposed code reduces the expected repair traffic by $37.0\%$--$54.8\%$ compared with conventional RS repair. The benefit is more pronounced for larger $r$, because the conjugate transformation significantly lowers the parity-node repair traffic while the piggybacking step continues to reduce data-node repair traffic.

\begin{table}[htpb]
\centering
\caption{Single-node repair-traffic simulation under uniform node failures. The ratios are normalized by $kr$ data symbols.}
\resizebox{\linewidth}{!}{
\begin{tabular}{|c|c|c|c|c|c|c|c|}
\hline
$k$ & $r$ & $L$ & $k/n$ & $\gamma_{\rm sys}$ & $\gamma_{\rm par}$ & $\gamma_{\rm all}$ & Reduction over RS \\
\hline
12 & 4 & 3 & 0.750 & 0.7014 & 0.4167 & 0.6302 & 37.0\% \\
24 & 4 & 3 & 0.857 & 0.6562 & 0.3854 & 0.6176 & 38.2\% \\
36 & 4 & 3 & 0.900 & 0.6412 & 0.3750 & 0.6146 & 38.5\% \\
52 & 4 & 3 & 0.929 & 0.6309 & 0.3702 & 0.6123 & 38.8\% \\
30 & 5 & 3 & 0.857 & 0.5978 & 0.3200 & 0.5581 & 44.2\% \\
36 & 6 & 3 & 0.857 & 0.5571 & 0.2731 & 0.5165 & 48.3\% \\
48 & 8 & 4 & 0.857 & 0.4922 & 0.2135 & 0.4524 & 54.8\% \\
\hline
\end{tabular}}
\label{tab:repair_simulation}
\end{table}

\section{Conclusion}
\label{sec:6}

In this paper, we proposed a new class of $(n,k,\ell)$ MDS array codes, called conjugate-piggybacking codes, with sub-packetization $\ell=n-k$. 
We showed that the proposed codes preserve the MDS property over moderate field sizes and achieve lower repair bandwidth than existing related piggybacking codes.
The repair-traffic simulation further shows that the proposed construction can substantially reduce the expected single-node repair traffic under representative high-rate storage parameters.
Future work includes constructing MDS array codes that further reduce repair bandwidth under small sub-packetization constraints.

\ifCLASSOPTIONcaptionsoff
\newpage
\fi

\bibliographystyle{IEEEtran}
\bibliography{CNC-v1}

@String { AUG          = {August} }

@String { COMPUTER     = {IEEE Computer Magazine} }

@String { GLOBECOM     = {Proc. IEEE GLOBECOM} }

@String { IEEETCOM     = {IEEE Trans. Communications} }

@String { IEEETIT      = {IEEE Trans. Information Theory} }

@String { INFOCOM      = {Proc. IEEE INFOCOM} }

@String { ISIT         = {Proc. {IEEE} Int. Symp. Inf. Theory} }

@String { MAY          = {May} }

@String { NOV          = {November} }

@String { SEP          = {September} }

@INPROCEEDINGS{2023gcpig,
  author={Jiang, Zhengyi and Shi, Hao and Huang, Zhongyi and Bai, Bo and Zhang, Gong and Hou, Hanxu},
  booktitle={GLOBECOM 2023 - 2023 IEEE Global Communications Conference}, 
  title={{Toward Lower Repair Bandwidth of Piggybacking Codes via Jointly Design for Both Data and Parity Nodes}}, 
  year={2023},
  volume={},
  number={},
  pages={7345-7350},
  doi={10.1109/GLOBECOM54140.2023.10436803}}

@ARTICLE{2018hetc,
  author={Kralevska, Katina and Gligoroski, Danilo and Jensen, Rune E. and Øverby, Harald},
  journal={IEEE Transactions on Big Data}, 
  title={{HashTag Erasure Codes: From Theory to Practice}}, 
  year={2018},
  volume={4},
  number={4},
  pages={516-529},
  doi={10.1109/TBDATA.2017.2749255}}

@INPROCEEDINGS{2023tran,
  author={Tang, Kaicheng and Cheng, Keyun and Chan, Helen H. W. and Li, Xiaolu and Lee, Patrick P. C. and Hu, Yuchong and Li, Jie and Wu, Ting-Yi},
  booktitle={IEEE INFOCOM 2023 - IEEE Conference on Computer Communications}, 
  title={{Balancing Repair Bandwidth and Sub-Packetization in Erasure-Coded Storage via Elastic Transformation}}, 
  year={2023},
  volume={},
  number={},
  pages={1-10},
  doi={10.1109/INFOCOM53939.2023.10228984}}

@INPROCEEDINGS{2023gcplus,
  author={Shi, Hao and Jiang, Zhengyi and Huang, Zhongyi and Bai, Bo and Zhang, Gong and Hou, Hanxu},
  booktitle={GLOBECOM 2023 - 2023 IEEE Global Communications Conference}, 
  title={{Piggybacking+ Codes: MDS Array Codes with Linear Sub-Packetization to Achieve Lower Repair Bandwidth}}, 
  year={2023},
  volume={},
  number={},
  pages={7351-7356},
  doi={10.1109/GLOBECOM54140.2023.10437215}}

@ARTICLE{2023wk1,
  author={Wang, Ke and Zhang, Zhifang},
  journal={IEEE Transactions on Communications}, 
  title={{Bidirectional Piggybacking Design for All Nodes With Sub-Packetization $2 \leq l \leq r$}}, 
  year={2023},
  volume={71},
  number={12},
  pages={6859-6869},
  doi={10.1109/TCOMM.2023.3311450}}

@INPROCEEDINGS{2023wk2,
  author={Wang, Ke and Zhang, Zhifang},
  booktitle={2023 IEEE Information Theory Workshop (ITW)}, 
  title={{Bidirectional Piggybacking Design for All Nodes with Sub-Packetization l = r}}, 
  year={2023},
  volume={},
  number={},
  pages={305-310},
  doi={10.1109/ITW55543.2023.10161644}}

@INPROCEEDINGS{2022Piggyback,
  author={Shi, Hao and Hou, Hanxu and Han, Yunghsiang S. and Lee, Patrick P. C. and Jiang, Zhengyi and Huang, Zhongyi and Bai, Bo},
  booktitle={2022 IEEE International Symposium on Information Theory (ISIT)}, 
  title={{New Piggybacking Codes with Lower Repair Bandwidth for Any Single-Node Failure}}, 
  year={2022},
  volume={},
  number={},
  pages={2601-2606},
  doi={10.1109/ISIT50566.2022.9834881}}

@article{jiang2024toward,
  title={{Toward Lower Repair Bandwidth and Optimal Repair Complexity of Piggybacking Codes with Small Sub-packetization}},
  author={Jiang, Zhengyi and Shi, Hao and Huang, Zhongyi and Bai, Bo and Zhang, Gong and Hou, Hanxu},
  journal={IEEE Transactions on Communications},
  year={2024},
  publisher={IEEE}
}

@INPROCEEDINGS{Shi2411:Conjugate,
AUTHOR="Hao Shi and Zhengyi Jiang and Zhongyi Huang and Bo Bai and Gong Zhang and
Hanxu Hou",
TITLE={{{Conjugate-Piggybacking} Codes: {MDS} Array Codes with Lower Repair
Bandwidth over Small Field Size}},
BOOKTITLE="2024 IEEE Information Theory Workshop (ITW) (ITW'2024)",
ADDRESS="Shenzhen, China",
PAGES="5.49",
DAYS=23,
MONTH=nov,
YEAR=2024
}

@misc{2023piggyback,
      title={{Two Piggybacking Codes with Flexible Sub-Packetization to Achieve Lower Repair Bandwidth}}, 
      author={Hao Shi and Zhengyi Jiang and Zhongyi Huang and Bo Bai and Hanxu Hou},
       journal = {arXiv e-prints},
         year = 2022,
        month = sep,
          eid = {arXiv:2209.09691},
        pages = {arXiv:2209.09691},
          doi = {10.48550/arXiv.2209.09691},
archivePrefix = {arXiv},
       eprint = {2209.09691},
 primaryClass = {cs.IT},
       adsurl = {https://ui.adsabs.harvard.edu/abs/2022arXiv220909691S}
}

@article{2017Piggybacking,
  title={{A Piggybacking Design Framework for Read-and Download-efficient Distributed Storage Codes}},
  author={ Rashmi, K. V.  and  Shah, Nihar B.  and  Ramchandran, Kannan },
  journal=IEEETIT,
  pages={5802-5820},
volume={63},  
number={9},
  year={2017}
}

@article{2018Repair,
  title={{A Repair-Efficient Coding for Distributed Storage Systems Under Piggybacking Framework}},
  author={ Yuan, Shuai  and  Huang, Qin  and  Wang, Zulin },
  journal=IEEETCOM,
  volume={66},
  number={8},
  pages={3245-3254},
  year={2018}
}

@article{2019AnEfficient,
  title={{An Efficient One-to-One Piggybacking Design for Distributed Storage Systems}},
  author={ Li, Gui Yang  and  Lin, Xing  and  Tang, Xiaohu },
  journal=IEEETCOM,
  volume={67},
  number={12},
  pages={8193-8205},
  year={2019}
}

@INPROCEEDINGS{2021piggyback,  
author={Zhengyi Jiang and Hanxu Hou and Yunghsiang S. Han and Zhongyi Huang and Bo Bai and Gong Zhang},  
booktitle= ISIT,   
title={An Efficient Piggybacking Design with Lower Repair Bandwidth and Lower Sub-packetization},
Year                     = {2021},
  Pages                    = {2328--2333}
}

@ARTICLE{2021piggybacking,  
author={Sun, Rong and Zhang, Lu and Liu, Jingwei},  
journal={IEEE Communications Letters},   
title={{A New Piggybacking Design with Low Repair Bandwidth and Complexity}},   
year={2021},  
volume={25},  
number={7},  
pages={2099--2103}}

@inproceedings{2018A,
  title={{A Tight Lower Bound on the Sub- Packetization Level of Optimal-Access MSR and MDS Codes}},
  author={ Balaji, S. B.  and  Kumar, P. Vijay },
  booktitle=ISIT,
   Year                     = {2018},
  Pages                    = {2381--2385}
}

@Article{dimakis2010,
  Title                    = {{Network Coding for Distributed Storage Systems}},
  Author                   = {Dimakis, A.G. and Godfrey, P.B. and Wu, Y. and Wainwright, M.J. and Ramchandran, K.},
  Journal                  = IEEETIT,
  Year                     = {2010},

  Month                    = {Sep.},
  Number                   = {9},
  Pages                    = {4539--4551},
  Volume                   = {56}
}

@Article{hou2019b,
  Title                    = {{Binary MDS Array Codes with Optimal Repair}},
  Author                   = {Hou, Hanxu and Lee, Patrick P C},
  Journal                  = IEEETIT,
  Year                     = {2020},
  Month                    = {Mar.},
  Number                   = {3},
  Pages                    = {1405–-1422},
  Volume                   = {66}
}

@Article{hou2019a,
  Title                    = {{Multi-Layer Transformed MDS Codes with Optimal Repair Access and Low Sub-Packetization}},
  Author                   = {Hou, Hanxu and Lee, Patrick P C and Han, Yunghsiang S},
  Journal                  = {arXiv preprint arXiv:1907.08938},
  Year                     = {2019}
}

@Article{hou2016,
  Title                    = {{BASIC Codes: Low-Complexity Regenerating Codes for Distributed Storage Systems}},
  Author                   = {Hou, Hanxu and Shum, Kenneth W. and Chen, Minghua and Li, Hui},
  Journal                  = IEEETIT,
  Year                     = {2016},
  Number                   = {6},
  Pages                    = {3053-3069},
  Volume                   = {62}
}

@Article{li2018,
  author    = {Li, Jie and Tang, Xiaohu and Tian, Chao},
  title     = {{A Generic Transformation to Enable Optimal Repair in MDS codes for Distributed Storage Systems}},
  journal   = IEEETIT,
  year      = {2018},
  volume    = {64},
  number    = {9},
  pages     = {6257--6267}
}

@Article{rashmi2011,
  Title                    = {{Optimal Exact-Regenerating Codes for Distributed Storage at the MSR and MBR Points via a Product-Matrix Construction}},
  Author                   = {K. V. Rashmi and N. B. Shah and P. V. Kumar},
  Journal                  = IEEETIT,
  Year                     = {2011},
  Month                    = aug,
  Number                   = {8},
  Pages                    = {5227--5239},
  Volume                   = {57}
}

@Article{reed1960,
  Title                    = {{Polynomial Codes over Certain Finite Fields}},
  Author                   = {Reed, Irving S and Solomon, Gustave},
  Journal                  = {Journal of the Society for Industrial \& Applied Mathematics},
  Year                     = {1960},
  Number                   = {2},
  Pages                    = {300--304},
  Volume                   = {8},
  Publisher                = {SIAM}
}

@Article{tamo2013,
  Title                    = {{Zigzag Codes: MDS Array Codes with Optimal Rebuilding}},
  Author                   = {Tamo, Itzhak and Wang, Zhiying and Bruck, Jehoshua},
  Journal                  = IEEETIT,
  Year                     = {2013},

  Month                    = may,
  Number                   = {3},
  Pages                    = {1597--1616},
  Volume                   = {59}
}

@article{2017Explicit,
  title={{Explicit Constructions of Optimal-Access MDS Codes with Nearly Optimal Sub-Packetization}},
  author={ Ye, Min  and  Barg, Alexander },
  journal={IEEE Transactions on Information Theory},
  volume={63},
  number={10},
  pages={6307--6317},
  year={2017}
}
\end{document}